%% file: topolMfg.tex
\documentclass[final,5p,twocolumn,10pt]{elsarticle} %SPM format with Times New Roman font

\include{preamble}

\journal{Symposium of Solid and Physical Modeling}

\begin{document}

\begin{frontmatter}

\title{A Classification of Topological Discrepancies in Additive Manufacturing}
 
%\author{Anonymous}     
\author{Morad Behandish, Amir M. Mirzendehdel, and Saigopal Nelaturi}
\address{\rm Palo Alto Research Center (PARC), 3333 Coyote Hill Road, Palo Alto,
	California 94304  \vspace{-15pt}}

\input{abstract}

\end{frontmatter}

%\linenumbers

\input{introduction}
\input{modeling}
\input{results}
\input{conclusions}

%\section*{References}

\bibliographystyle{elsarticle-num} 
\bibliography{topolMfg}

\end{document}

%% file: preamble.tex
\usepackage{lineno}
\usepackage{hyperref}
\usepackage{enumitem}
\modulolinenumbers[5]

\usepackage{amsmath}
\usepackage{amssymb}
\usepackage{stmaryrd}
\usepackage{amsthm}
\usepackage{amsfonts}
\usepackage{dsfont}
\usepackage{graphicx}
\usepackage{upgreek}
\usepackage{bm}
\usepackage{color}
\usepackage{float}
\usepackage{tikz-cd}

\usepackage{todonotes}

\biboptions{numbers,sort&compress}

\newcommand{\bx}{{\mathbf{x}}}

\newcommand{\bt}{{\mathbf{t}}}

\newcommand{\sign}{\mathsf{sign}}

\newcommand{\indic}{{\mathbf{1}}}

\newcommand{\dimm}{\mathrm{d}}

\newcommand{\sweep}{\mathsf{sweep}}

\newcommand{\R}{{\mathds{R}}}

\newcommand{\SO}[1]{{\mathrm{SO}(#1)}}
\newcommand{\SE}[1]{{\mathrm{SE}(#1)}}

\newcommand{\supp}{\mathsf{supp}}

\newcommand{\asM}{\mathsf{M}}
\newcommand{\asD}{\mathsf{D}}

\newcommand{\conf}{{\mathsf{C}}}

\newcommand{\motion}{T}
\newcommand{\mmf}{B}

\newcommand{\ec}{\chi}
\newcommand{\ecc}{\hat{\ec}}
\newcommand{\beti}{\beta}

\renewcommand{\th}{$^\text{th}$ }

\theoremstyle{definition}

\newtheorem{defn}{Definition}

\newtheorem{lemma}{Lemma}
\newtheorem{coro}{Corollary}

\newcommand{\parag}[1]{\paragraph{\bf {#1}}}

\newcommand{\eq}[1]{(\ref{#1})} %Use \eq{...} for referincing equations as (...)
\newcommand{\com}[1]{} %Use \com{...} to commenting out multiple paragraphs

%% file: abstract.tex
\begin{abstract}

Additive manufacturing (AM) enables enormous freedom for design of complex
structures. However, the process-dependent limitations that result in
discrepancies between as-designed and as-manufactured shapes are not fully
understood. The tradeoffs between infinitely many different ways to approximate
a design by a manufacturable replica are even harder to characterize. To support
design for AM (DfAM), one has to quantify local discrepancies introduced by AM
processes, identify the detrimental deviations (if any) to the original design
intent, and prescribe modifications to the design and/or process parameters to
countervail their effects. Our focus in this work will be on topological
analysis. There is ample evidence in many applications that preserving local
topology (e.g., connectivity of beams in a lattice) is important even when
slight geometric deviations can be tolerated. We first present a generic method
to characterize local topological discrepancies due to material under- and
over-deposition in AM, and show how it captures various types of defects in the
as-manufactured structures. We use this information to systematically modify the
as-manufactured outcomes within the limitations of available 3D printer
resolution(s), which often comes at the expense of introducing more geometric
deviations (e.g., thickening a beam to avoid disconnection). We validate the
effectiveness of the method on 3D examples with nontrivial topologies such as
lattice structures and foams.

\end{abstract}


%% file: introduction.tex
\section{Motivation}

Additive manufacturing (AM) has overcome many of the limitations imposed on
design by traditional formative and subtractive manufacturing processes. It
enables making complex light-weight parts with superior structural
\cite{Robbins2016efficient} or thermal \cite{Haertel2017fully} performance,
among other functional benefits. The added flexibility and complexity, however,
come with a new set of unsolved computational challenges. Often, the
\textit{as-manufactured} structures differ from the \textit{as-designed} targets
in ways that are harder to characterize, quantify, and correct compared to
traditional processes such as machining \cite{Balic2006intelligent}. The
deviations depend not only on the geometric and material properties of the part,
but also on pre-processing (e.g., slicing and path-planning) and AM process
parameters such as build direction, resolution, and temperature. The
manufacturing limitations become evident only after the build process is
finished \cite{Nouri2016structural}.

In many cases, the deviations of the as-manufactured structure from the
as-designed target are not completely preventable. However, the deviations can
be controlled, by slight modification of the design or process, to prevent
unpredictable failures. Even when an as-designed target is not manufacturable
as-is, {\it there are infinitely many manufacturable alternatives that closely
	approximate its shape,} while preserving the design intent and function. For
example, if the as-designed shape has smaller features than the 3D printer's
resolution, there may be many plausible ways to thicken its small structural or
aesthetic features with no compromise in its form, fit, and function.

In this paper, we focus on preserving shape properties, particularly those
pertaining to {\it topological integrity} of the design. For many AM structures
(e.g., infill lattices \cite{Wu2016self,Wu2018infill} and foams
\cite{Martinez2016procedural,Martinez2017orthotropic}), the topological
integrity of the structure has substantial functional significance---commonly
even more important than geometric precision. Often times, ``small'' enough
deviations (from a metric standpoint) from the as-designed geometry may lead to
changes in topology that compromise performance. For example, if the shape of a
lattice structure is slightly deformed due to the stair-stepping effect of
layered fabrication, it may not matter as much as preserving its connectivity.
Common examples of functional failure are broken beams in load-bearing
micro-structures, covered tunnels in heat exchanger micro-channels, filled
cavities in porous meta-materials, and so on. Topological properties can impact
manufacturing post-processing as well, such as powder removal in SLS or DMLS.

We present a novel approach to detect and classify local contributions to
topological discrepancies to enable surgical corrections that preserve the
global and local topology.

\subsection{Related Work}

There are few computational tools for predicting and/or alleviating shape
defects introduced by AM. Most of the methods that study the quality of AM
processes use geometric measures such as volumetric error
\cite{Cajal2013volumetric,Tong2008error,Telea2011voxel}, dimensional accuracy
\cite{Ibrahim2009dimensional,Rao2016assessment,Silva2008dimensional}, or surface
finish \cite{wang2016improved} as the main criteria. In other words, feasibility
of an AM process for a given design is mainly assessed either through visual and
qualitative evaluation \cite{Mahesh2004benchmarking} or quantitative and
statistical analysis on feature sizes and tolerances
\cite{Arrieta2012quantitative,Cabiddu2017maps}. For instance
\cite{Nelaturi2015manufacturability,Nelaturi2015representation} demonstrated
methods to identify, visualize, and correct thin features in 3D printed parts.
However, they offer limited insight on topological consequences of deviations
from the original design.

Despite its significance, topology-aware DfAM has received less attention than
deserved. For example, a recent study \cite{Rosen2018inferring} proposed using
methods from topological data analysis such as mapper graphs and persistent
homology \cite{Edelsbrunner2008persistent} to assess the quality of point cloud
data for 3D printing along a given build direction. The method is focused on
data quality and representation rather than process effects. Another study
\cite{Liu2015identification,Li2016structural} proposed a method to restrict
enclosed voids in topology optimization. The objective function is penalized by
simulated results of artificial heat transfer with large conductivity assigned
to the voids. The heuristics are effective in alleviating cavity formations but
provide no guarantees. In \cite{Nelaturi2015manufacturability}, the authors used
medial axes to separate and thicken thin features. Medial axes are hard to
compute due to their instability in the presence of noise
\cite{Attali2009stability}. In \cite{Cabiddu2017maps}, the authors used a
different criterion for global thin feature detection based on topological
effects of $\epsilon-$ball inclusion/exclusion. Most existing methods use
heuristics for thickening of thin features, rarely consider other classes of
features (e.g., tunnels/cavities), and do not quantify the local topological
errors in a way that can be used in trade-off with other criteria.

\begin{figure}
	\centering
	\includegraphics[width=0.46\textwidth]{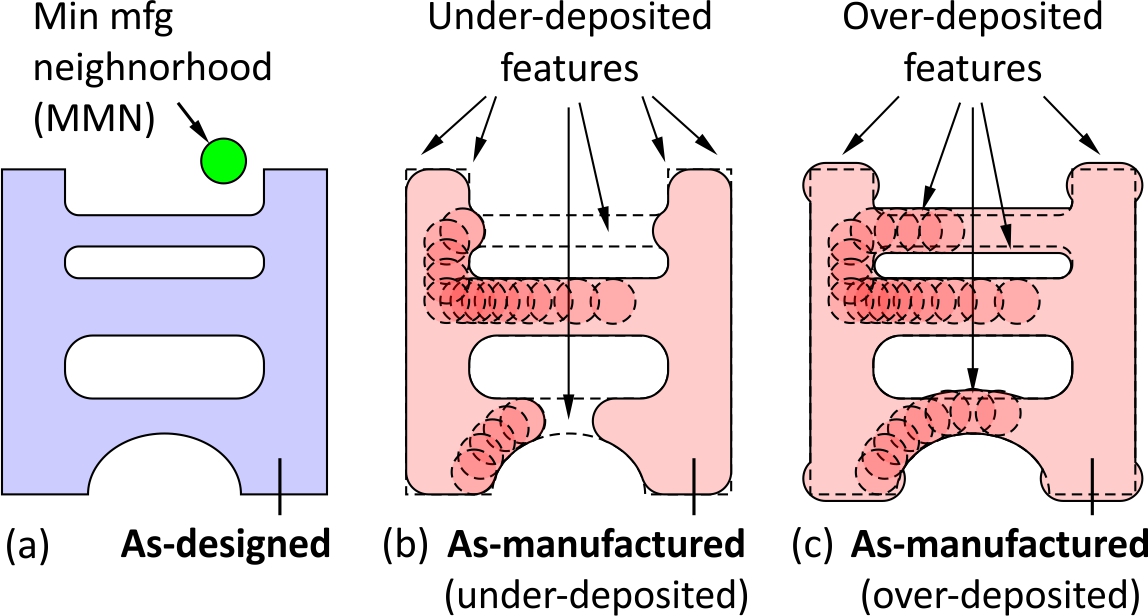}
	\caption{Consider one 2D slice in 3D printing with an (exaggerated) minimum
	feature size constraint. Different approximations to the as-designed slice can
	be obtained in the trade-space of desired properties: (a) losing connections
	and holes to remain within the boundaries; (b) preserving topology at the
	expense of over-depositing.} \label{fig_idea}
\end{figure}

Our goal is to distinguish features (of arbitrary shape) that adversely affect
the topology from those that do not. For instance, a missing bridge between two
connected components in the as-manufactured shape is qualitatively different
from a sharp corner being rounded off without affecting connectivity (Fig.
\ref{fig_idea}). The former changes the local topology, even if the global
topology (e.g., total number of connected components) is intact due to possible
existence of other connections. Hence, global topological analysis based on Reeb
graphs, mapper graphs, persistent homology, and other tools of the trade can
miss them.

\subsection{Contributions \& Outline}

This paper presents a novel computational framework to quantify and correct 
topological discrepancies for AM parts; specifically the contributions of the 
paper are:
\begin{enumerate}
	\item presenting parameterized families of manufacturable alternatives that
	closely approximate a given design of arbitrary shape for a given AM
	capability;
	\item quantifying as-manufactured deviations from the as-designed model in
	terms of under- and over-deposition (UD/OD) regions for every such alternative;
	\item characterizing the local contributions of the deviant UD/OD components to
	the global topological discrepancies in terms of Euler characteristic (EC); and
 	\item prescribing a systematic process to locally adjust the AM instrument
	motions to counteract the topological deviations with minimal geometric
	alteration.
\end{enumerate}
In Section \ref{sec_model}, we extend previous work on measure-theoretic
parametrization of as-manufactured models for a given as-designed model and AM
capability \cite{Nelaturi2015manufacturability,Nelaturi2015representation,
	Behandish2018automated}.

\begin{figure*}
	\centering
	\includegraphics[width=0.96\textwidth]{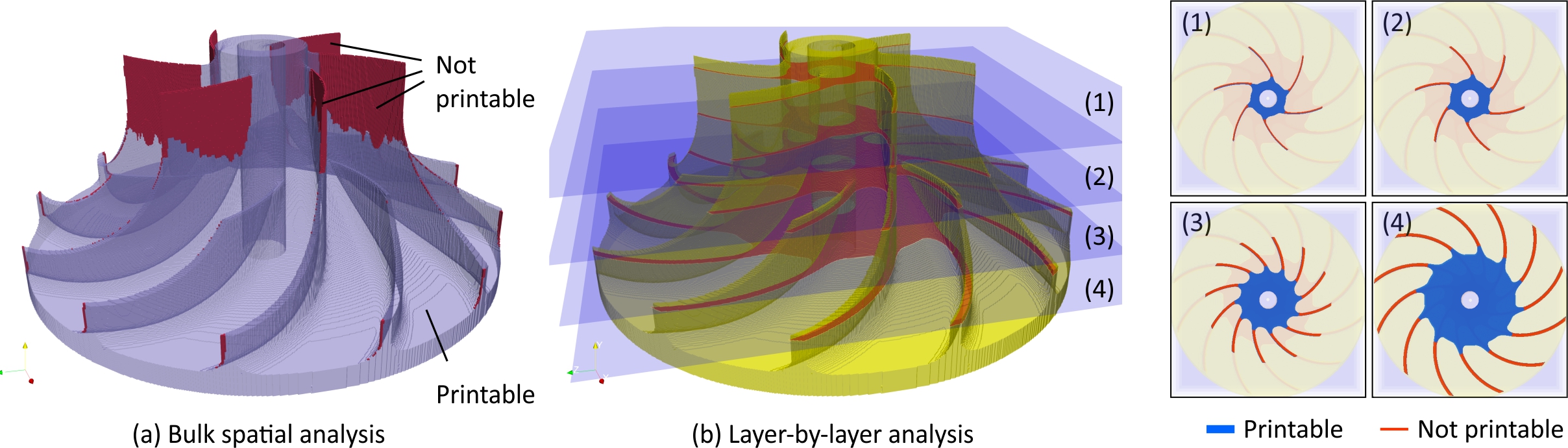}
	\caption{AM analysis can be performed in 3D (a) to obtain a first approximation
	based on feature size without considering build direction, support structures,
	etc. A more precise analysis can be done in 2D a slice-by-slice basis (b) after
	the build direction is selected. The latter also has the benefit of allowing
	much higher resolutions due to lower dimensions. The most accurate analysis is
	done after the precise 1D trajectories are selected and parameterized (e.g.,
	G-code). An exaggerated MMN size is used to make the under-depositions obvious.
	In this case, there are no topological changes due to the non-printable regions
	on the blades. The 3D model is obtained from \textsf{GrabCAD}.} \label{fig_methods}
\end{figure*}

In Section \ref{sec_UDOD}, we present a novel approach of comparative
topological analysis (CTA) for a pair of arbitrary shapes. We decompose the
inevitable discrepancies between the as-designed and as-manufactured shapes due
to process limitations into UD/OD ``features.'' We present a novel computational
approach to quantify their local contributions to the global topological
changes. Unlike global quantifiers such as the EC and Betti numbers (BN) that
count the total number of connected components, tunnels, and cavities in 3D, our
analysis provides local and precise {\it spatial information} about the
comparative topology of the as-designed and as-manufactured shapes.

A theoretical contribution with major computational utility is a general formula
to {\it locally quantify topological differences between two arbitrary
	overlapping shapes}. Using the additivity of EC, we define the fundamental
notion of (topological) {\it simplicity} to decompose and classify the set
differences between two arbitrary pointsets.

The classification of local discrepancies is used for surgical modification of
the design and process to alleviate undesirable topological discrepancies. We
propose a simple approach to locally update the deposition policy while
retaining the minimum feature size constraints. The approach can be applied
recursively to find the ``best'' manufacturable approximation to the target
design---among many alternatives, as exemplified in Fig. \ref{fig_idea}---that
preserves its topological integrity with the least necessary geometric
deviation. However, topological integrity is but one of many (potentially
competing) criteria for multi-objective design space exploration
\cite{Nelaturi2019exploring} and has to be satisfied/optimized for the best
trade-offs. Design space exploration is beyond the scope of this paper.

%% file: modeling.tex
\section{Geometric Modeling for AM} \label{sec_model}

In this section, we present a general approach for modeling manufacturable
alternatives that closely approximate a given as-designed target of arbitrary
shape, while complying with geometric limitations of AM (Section
\ref{sec_general}). Our model offers immense flexibility by generating {\it
	families} of as-manufactured variants over a range of custom AM parameters such
as resolution and deposition allowance, rather than a single model (Section
\ref{sec_family}). We conclude this section by a discussion of future directions
for extending this model to {\it spatially varying} allowance fields to provide
additional freedom for local adjustments (Section \ref{sec_model_local}).

\subsection{From As-Designed to As-Manufactured} \label{sec_general}

Practical limitations on the AM resolution and wall thickness are the source of
inevitable geometric and topological discrepancies between the as-designed
target and as-manufactured outcome. For example, the layer thickness in direct
metal laser sintering (DMLS) is typically in the range 0.3--0.5 mm along the
deposition layer and 20--30 microns along the build direction at the highest
resolution, with a minimum hole diameter of 0.90--1.15 mm within a typical
workspace of $250 \times 250 \times 325$ mm$^3$. Attempting to fabricate designs
that have smaller features will result in disconnected beams, filled
holes/tunnels, or hard-to-predict combinations. The as-manufactured part may
eventually look nothing like the as-designed target.

As a first approximation, manufacturability analysis can be formulated as a
purely geometric problem. Earlier work
\cite{Nelaturi2015manufacturability,Nelaturi2015representation} has developed
methods to model as-manufactured structures using basic notions from group
morphology \cite{Lysenko2010group}. Similar approaches have also been used to
generate hybrid (interleaved additive and subtractive) manufacturing process
plans \cite{Behandish2018automated}. We briefly review the relevant elements of
these methods used in this paper.

An as-manufactured shape is obtained by sweeping a minimum manufacturable
neighborhood (MMN) of arbitrary shape along an arbitrary motion that is allowed
by the machine's degrees of freedom (DOF) \cite{Behandish2018automated}. For
example, most 3D printers operate by 3D translations of a printer head (no
rotational DOF) over the workpiece as they deposit a droplet of material whose
shape is represented by the MMN (e.g., an ellipsoid/cylindroid). The MMN's size
(e.g., diameters/height) and orientation are determined by the printer
resolution and the build direction. Unless the as-designed shape is perfectly
sweepable by the MMN via an allowable motion, the as-manufactured shape will
have to differ. The challenge is to find the ``best'' allowable motion of the
deposition head whose sweep of MMN results in a shape as close as possible to
the as-designed target. The answer is not unique, as it depends on the notion of
closeness; i.e., the criteria based on which the discrepancies are measured. One
such criterion comes from local topological considerations that impact function
(Section \ref{sec_UDOD}).

\subsection{Computing a Family of As-Manufactured Models} \label{sec_family}

For ease of illustration, the examples of this paper are restricted to AM
instruments with translational DOF only, i.e., the machine does not allow
rotations of the deposition head. This assumption covers commercial 3D printers
that deposit flat layers on top of each other. After presenting the results for
the simpler case of translations, we also discuss the more general case of
general rigid motions supported by high-axis CNC machines.

\parag{Manufacturability Measures}

Throughout this paper, we adopt a {\it measure-theoretic} approach to define and
compute as-manufactured shapes. At every point in the 3D space inside the
printer workspace---which represents a hypothetical translational configuration
of the printer head---we quantify manufacturability by computing the {\it
	overlap measure} (OM) between the stationary as-designed target and a
hypothetical MMN instance displaced to the query point. For example, the
analysis for 3D printing on flat layers with translational DOF can be performed
in at least two ways, as illustrated in Fig. \ref{fig_methods}:
\begin{enumerate}
	\item {\bf Bulk spatial analysis:} A 3D field of overlap measures is obtained
	between the 3D as-designed model and a 3D model of the MMN, e.g., a blob of
	material that is representative of a deposition unit. The layer thickness and
	build orientation may or may not be encoded into the shape of the MMN. The
	measure is the volume of intersections between 3D shapes.
	\item {\bf Layer-by-layer analysis:} The as-designed model is sliced along a
	fixed build direction into layers that are a constant distance apart, e.g.,
	obtained from a 3D printer's known layer thickness specs. For each 2D
	as-designed slice, a 2D field of overlap measure is constructed by using a 2D
	model of the MMN, e.g., nozzle or laser beam cross-section. The measure is the
	surface area of intersections between 2D shapes.
\end{enumerate}
(1) provides a rapid approximation to the latter and allows a high-level
analysis in the absence of slicing parameters. It can be computed in one shot
for digitized translational motions by a convolution of indicator (i.e.,
characteristic) functions of the as-designed shape and MMN in 3D
\cite{Nelaturi2015representation}. On a commodity computer, it can be computed
for resolutions of $\sim 10^2$--$10^3$ voxels per coordinate axis due to storage
limitations. (2) provides a mid-level analysis that takes build direction and
slicing into account, but ignores 1D tool path, digitization (e.g., G-code), and
other machine-level details. It can be scaled to resolutions as high as $\sim
10^5$ pixels per coordinate axis when each slice is stored and analyzed
separately as a bitmap image. For general rigid motions, it is replaced by a
parametric motion over a 2D curvilinear surface.

The drawback of (2) is that preserving topological properties per slice does not
guarantee preserving topological properties in 3D, which is what matters from a
functional standpoint. We will use 2D examples for illustration purposes, and
demonstrate 3D examples in Section \ref{sec_results}.

\parag{Computing AM Instrument Motions}

We represent both cases uniformly in terms of Lebesgue $\dimm-$measures
$\mu^\dimm[\cdot]$ ($\dimm = 2$ or $3$) of intersection between the as-designed
shape/slice, denoted by $\Omega_\asD \subseteq \R^\dimm$ and MMN (before
displacement), denoted by $\mmf \subseteq \R^\dimm$. These pointsets are assumed
to be solids (i.e., `r-sets') \cite{Requicha1980representations} defined as
compact (bounded and closed) regular and semianalytic subsets of the
$\dimm-$space, which automatically deems them {\it $\dimm-$measurable}. We
define the OM as a real-valued field over the {\it configuration space}
($\conf-$space) of the AM instrument. At a given {\it configuration} $\tau \in
\conf$, the MMN is hypothetically displaced to $\tau \mmf = \{\tau \bx ~|~ \bx
\in \mmf \}$, where $\tau \bx \in \R^\dimm$ stands for a rigid transformation of
$\bx \in \R^\dimm$. The OM is thus given by $f_\mathsf{OM}(\tau) :=
\mu^\dimm[\Omega_\asD \cap \tau \mmf]$.

For example, for $3-$axis machines with no rotational DOF (i.e., $\conf \cong
\R^3$), the configurations can be represented by $\dimm-$dimensional translation
vectors $\bt \in \R^3$. The displaced MMN is thus obtained as $(\mmf + \bt) =
\{\bx + \bt ~|~ \bx \in \mmf \}$. For high-axis machines with rotations, the
instrument configurations are elements of a different subgroup of the 6D Lie
group of rigid motions $\SE{3}$ \cite{Selig2005geometrical}, hence the OM is a
field defined over a non-Euclidean (Riemannian) manifold. For simplicity, we
present the mathematical formula for $\conf \cong \R^\dimm$ (translations only)
first and briefly discuss extensions, whenever possible, to $\conf = \SE{\dimm}
\cong \SO{3} \rtimes \R^\dimm$ (combined rotations and translations).

The OM field varies from $f_\mathsf{OM}(\tau) = 0$ (no overlap) to the total
measure of the MMN $f_\mathsf{OM}(\tau) = \mu^\dimm[\mmf]$ (full overlap), the
latter corresponding to the configurations (if any) at which the displaced MMN
is completely contained within the as-designed shape (i.e., $\tau \mmf \subseteq
\Omega_\asD$). The OM ratio (OMR) $f_\mathsf{OMR}(\tau) := f_\mathsf{OM}(\tau)/
\mu^\dimm[\mmf] \in [0, 1]$ thus defines a normalized continuous field over the
$\conf-$space. The idea is illustrated in Fig. \ref{fig_OMR2D} for a 2D example.

\begin{figure} [ht!]
	\centering
	\includegraphics[width=0.46\textwidth]{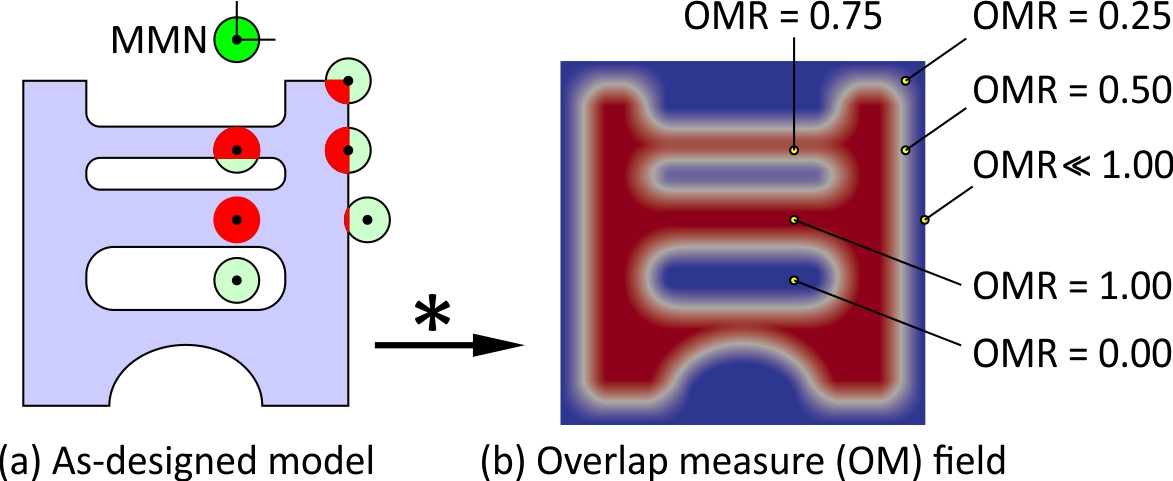}
	\caption{For 2D translations ($\conf \cong \R^2$) the OMR at different query
	points is computed by moving the MMN to the query point and computing the
	overlap $2-$measure (i.e., surface area) between the moved MMN and the
	as-designed slice.} \label{fig_OMR2D}
\end{figure}

Let $0 \leq \lambda < 1$ represent a parametrization that continuously connects
the two extremes. We define $T_\lambda \subseteq \conf$ as a superlevel set of
the OM field:
\begin{equation}
	\motion_\lambda := \big\{ \tau \in \conf ~|~ \mu^\dimm[\Omega_\asD \cap \tau
	\mmf] > \lambda \mu^\dimm[\mmf] \big\}, \label{eq_family}
\end{equation}
The members of the family of motions $\{\motion_\lambda\}_{0 \leq \lambda <
	1}$ are distinguished by the requirement that for every rigid transformation
$\tau \in T_\lambda$, the intersection measure between the stationary
as-designed shape and the displaced MMN is lower-bounded by a $\lambda-$fraction
of the maximum possible overlap (Fig. \ref{fig_OMR2D}). Hereafter, we refer to
this fraction as the OMR {\it threshold} (OMRT).

\begin{figure*}
	\centering
	\includegraphics[width=0.96\textwidth]{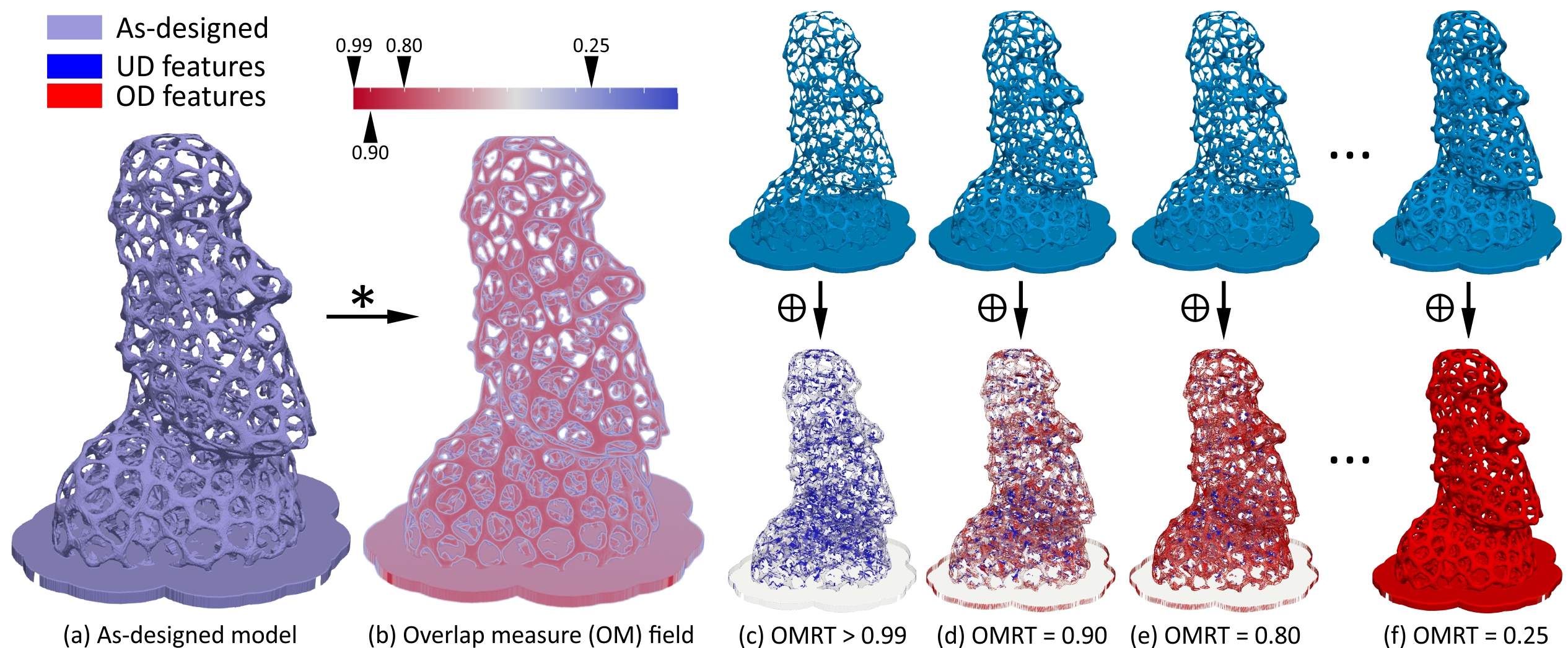}
	\caption{For a given as-designed model (a) in 3D (voxelized at $512
	\times 512 \times 512$) its OM field (b) with an MMN (diamond-shaped, with
	edge-length of $3\times 3 \times 3$) is computing as a convolution, using FFTs
	on the GPU. Different superlevel sets of the OM field (c--f) represent 3D
	translations of the AM instrument (top) for various deposition allowances
	specified by OMRT. The dilation of these motions with the MMN yield the
	as-manufactured alternatives (bottom). For each one, the UD/OD regions are
	color-coded by blue/red, respectively. The 3D model is obtained from
	\textsf{GrabCAD}.} \label{fig_OMR3D}
\end{figure*}

At every choice of the OMRT, all configurations whose OMR values exceed the
$\lambda-$fraction are included in the motion. For purely translational motions
(i.e., $\conf \cong \R^\dimm$), the near-extreme superlevel sets corresponding
to $\lambda \to 1^-, 0^+$ can be computed by Minkowski difference/sum,
respectively,%
\footnote{We are ignoring the technicalities with regard to regularization of
	superlevel sets. All equalities are equivalence up to regularization (i.e.,
	equality ``almost everywhere'' in measure-theoretic sense)
	\cite{Behandish2017analytic}.}
of the as-designed shape with the MMN's reflection with respect to the origin
\cite{Nelaturi2015manufacturability}:
\begin{align}
	T_{1^-} := \lim_{\lambda \to 1^-}\motion_{\lambda} &= (\Omega_\asD \ominus
	(-\mmf)), \label{eq_lambda_min} \\
	T_{0^+} := \lim_{\lambda \to 0^+}\motion_{\lambda} &= (\Omega_\asD \oplus
	(-\mmf)). \label{eq_lambda_max}
\end{align}
where $-B = \{ -\bx ~|~ \bx \in B \}$ denotes a reflection. For other values
$\lambda \in [0,1)$, $(\Omega_\asD \ominus (-\mmf)) \subseteq \motion_\lambda
\subseteq (\Omega_\asD \oplus (-\mmf))$, as depicted by Fig. \ref{fig_OMR2D}. In
other words, the one-parametric family of as-manufactured alternatives form a
totally ordered set (in terms of containment) bounded by the two extremes given
in \eq{eq_lambda_min} and \eq{eq_lambda_max}.

Hereafter, we frequently characterize sets with their {\it defining functions}.
A real-valued defining function implicitly describes a set as its $0-$superlevel
set (i.e., `support'): For example, $f_{T_\lambda} : \conf \to \R$ describes
$T_\lambda \subseteq \conf$ as:
\begin{equation}
	\supp(f_{T_\lambda}) := \big\{ \bt \in \conf ~|~ f_{T_\lambda}(\bt) > 0 \big\}.
	\label{eq_supp}
\end{equation}
where $f_{T_\lambda}(\tau)$ for a given $\tau \in \conf$ is defined by:
\begin{equation}
	f_{T_\lambda}(\tau) := \mu^\dimm[\Omega_\asD \cap \tau \mmf] - \lambda
	\mu^\dimm[\mmf], \label{eq_def_motion}
\end{equation}
whose substitution into \eq{eq_supp} yields the definition in \eq{eq_family}.
The indicator (i.e., characteristic) function $\indic_{T_\lambda} : \conf \to
\{0, 1\}$ is a special canonical form of a defining function:
\begin{equation}
	\indic_{T_\lambda}(\tau) := \left\{
	\begin{array}{ll}
		1 & \text{if}~ \mu^\dimm[\Omega_\asD \cap 
		\tau \mmf] > \lambda \mu^\dimm[\mmf], \\
		0 & \text{otherwise}.
	\end{array}
	\right. \label{eq_indic_motion}
\end{equation}
In general, $\indic_{T_\lambda}(\tau) = \sign (f_{T_\lambda} (\tau))$ where
$\sign(x) = 1$ for $x > 0$ and $\sign(x) = 0$ otherwise.

Next, we transform the above measures into operations on indicator functions of
the as-designed shape and MMN. Intersection measures can in general be computed
as inner products of indicator functions. If the pointsets are shifted by
different relative transformations $\bt \in \R^3$, the inner product for all
translations can be expressed as a single {\it convolution}
\cite{Lysenko2010group}. Hence, for translational motions (i.e., $\conf \cong
\R^\dimm$), the OM can be computed as:
\begin{equation}
	\mu^\dimm[\Omega_\asD \cap (\mmf + \bt)] %= \langle \indic_{\Omega_\asD},
	% \indic_{\tau \mmf} \rangle
	= (\indic_{\Omega_\asD} \ast \indic_{-\mmf})(\bt).
\end{equation}
The MMN measure is simply obtained as a $1-$norm of the indicator function
$\mu^\dimm[\mmf] = \| \indic_\mmf \|_1$. Hence, we rewrite \eq{eq_def_motion} in
{\it exact} analytical form (i.e., in terms of functions):
\begin{equation}
	f_{T_\lambda}(\bt) = (\indic_{\Omega_\asD} \ast \indic_{-\mmf})(\bt) - \lambda
	\|\indic_\mmf\|_1, \label{eq_def_motion_conv}
\end{equation}
The main advantage of this formulation is its computational efficiency.
Convolutions can be computed rapidly via fast Fourier transforms (FFT)
\cite{Kavraki1995computation}. For a grid of size $O(n)$ per dimension, this
takes $O(n^\dimm \log n)$, which is very close to linear time in the number of
pixels/voxels $O(n^\dimm)$ as opposed to $O(n^{2\dimm})$ of the na\"{i}ve
approach. When rotations are involved, the above formula holds for every fixed
rotation in the $\conf-$space \cite{Lysenko2010group}, hence the computation is
repeated for a sampling of rotations.

For general rigid motions, the above notions are subsumed by Minkowski
products/quotients and group convolutions \cite{Lysenko2010group}. The general
formula is slightly different due to compositions with lifting/projection maps
between the Euclidean $\dimm-$space and the Lie group $\conf = \SE{\dimm} \cong
\SO{\dimm} \rtimes \R^\dimm$. The general idea is nevertheless applicable.

\parag{Computing As-Manufactured Shapes}

For each set of AM instrument configurations $T_\lambda \subseteq \conf$, the
resulting as-manufactured shape $\Omega_{\asM, \lambda} \subseteq \R^\dimm$ is
obtained by sweeping the MMN along $T_\lambda$. For purely translational motions
(i.e., $\conf \cong \R^\dimm$) the sweep reduces to a Minkowski sum:
\begin{equation}
	\Omega_{\asM, \lambda} := \sweep(T_\lambda, B) \overset{\conf \cong
	\R^\dimm}{=\joinrel=\joinrel=} (T_\lambda \oplus B), \label{eq_sweep}
\end{equation}
which extends to Minkowski product for general rigid motion. The defining
function of Minkowski sum is also computable as a convolution of indicator
functions \cite{Lysenko2010group}:
\begin{equation}
	f_{\Omega_{\asM, \lambda}}(\bx) = f_{T_\lambda \oplus \mmf}(\bx) =
	(\indic_{T_\lambda} \ast \indic_{\mmf})(\bx), \label{eq_Mink_conv}
\end{equation}
where the indicator function of $T_\lambda$ is obtained from
\eq{eq_indic_motion}. The as-manufactured shape is obtained as the support of
the convolution $\Omega_\asM = \supp(\indic_{T_\lambda} \ast \indic_{\mmf})$.
Once again, the convolution takes $O(n^\dimm \log n)$ using FFTs, which is only
a logarithmic factor away from linear time.

In summary, a one-parametric family of AM instrument motions are computed (for
$\conf \cong \R^\dimm$) as:
\begin{equation}
	~\motion_{\lambda} = \big\{ \bt \in \conf ~|~ (\indic_{\Omega_\asD} \ast
	\indic_{-\mmf})(\bt) - \lambda \|\indic_\mmf\|_1 > 0\big\}.
	\label{eq_family_motion}
\end{equation}
Every $T_\lambda$ is distinguished by its {\it uniform} OMRT value of $\lambda
\in [0, 1)$ throughout the $\conf-$space. The as-manufactured shape of these
motions are then computed as:
\begin{equation}
	\Omega_{\asM, \lambda} = \big\{ \bx \in \R^\dimm ~|~ ((\sign \circ
	f_{T_\lambda}) \ast \indic_{\mmf})(\bx) > 0 \big\}. \label{eq_family_shape}
\end{equation}

\subsection{Modeling Local Geometric Modifications} \label{sec_model_local}

The one-parametric family of as-manufactured alternatives form a 
totally ordered set (in terms of containment) bounded by the two extremes; 
namely,
\begin{enumerate}
	\item minimized under-deposition (UD$^-$) ($\lambda \to 1^-$) in which the 
resulting as-manufactured shape is the unique maximal (in terms of containment) 
manufacturable shape contained within the as-designed shape; and
	\item conservative over-deposition (OD$^+$) ($\lambda \to 0^+$) in which the 
resulting as-manufactured shape is a generalized offset of the as-designed 
shape with the MMN, and contains the MMN with a conservative margin.
\end{enumerate}
For translational motions, the UD$^-$ shape is a morphological opening (i.e.,
dilation of erosion) while the OD$^+$ shape is a double-offset (i.e., dilation of
dilation), obtained by applying \eq{eq_sweep} to \eq{eq_lambda_min} and
\eq{eq_lambda_max}, respectively:
\begin{align}
	\Omega_{\asM, 1^-} := \lim_{\lambda \to 1^-} \Omega_{\asM, \lambda} &=
	(\Omega_\asD \ominus (-\mmf)) \oplus \mmf, \label{eq_lambda_min_asM} \\
	\Omega_{\asM, 0^+} := \lim_{\lambda \to 0^+} \Omega_{\asM, \lambda}  &=
	(\Omega_\asD \oplus (-\mmf)) \oplus \mmf. \label{eq_lambda_max_asM}
\end{align}
Figures \ref{fig_OMR2D} and \ref{fig_OMR3D} illustrate the extreme cases as
well as several possibilities in between them, in 2D and 3D, respectively.
Decreasing the OMR leads to uniform global thickening of $\Omega_{\asM,
	\lambda}$ that depends on the local geometry of $\Omega_\asD$. If the MMN is
small, all $\Omega_{\asM, \lambda}$ alternatives have small geometric deviations
from $\Omega_\asD$. Nevertheless, they can have dramatically different
topological properties (Fig. \ref{fig_OMR3D}) as we discuss in Section
\ref{sec_UDOD}.

The choice of OMRT depends on trade-offs between functional requirements. If
over-deposition is strictly prohibited, $\Omega_{\asM, 1^-}$ is an obvious
choice if the goal is to minimize geometric difference. However, there will be
topological consequences due to missing UD features that cause local
disconnections (Fig. \ref{fig_OMR3D} (c)). If there is some allowance for OD, it
is possible to bring those features back by decreasing $\lambda$ (Fig.
\ref{fig_OMR3D} (d--f)). However, the resulting {\it uniform} growth of the
as-manufactured shape can cause topological errors elsewhere in the shape, such
as covering a tunnel/cavity. The geometric allowance may not be uniform either,
as functional surfaces may not be grown as much as aesthetic ones. It is
possible to extend the above model to accommodate adaptive, non-uniform, and
local changes to the as-manufactured shape.

To enable more design freedom, a {\it non-uniform} OMRT can be defined as a
field $\lambda^\ast: \conf \to \R$ and parameterized using a set of basis
functions:
\begin{equation}
	\lambda^\ast(\tau) := \sum_{0 \leq j < m} \lambda_j \phi_j(\tau),
	\label{eq_lambda_field}
\end{equation}
where $\phi_j : \conf \to \R$ can be any basis that provides flexible local
adjustments (e.g., radial basis functions or splines). The coefficients
$\lambda_0, \lambda_1, \ldots, \lambda_m \in [0, 1)$ can be adjusted (e.g.,
using gradient-descent optimization) based on the feedback from geometric,
topological, or physical analysis of the as-manufactured model. Without loss of
generality, we let $\phi_0(\tau) := 1$ so that the extended model subsumes the
uniform OMRT model as a special case when $m := 1$. The extended model is
obtained by replacing $\lambda$ with $\lambda^\ast(\tau)$ in the defining
functions of Section \ref{sec_family}.

Since topological analysis is the main focus of this paper, we postpone a more
detailed discussion of as-manufactured shape modeling to future work.

\section{Comparative Topological Analysis} \label{sec_UDOD}

We present a novel approach to characterize the differences in basic topological
properties of an arbitrary as-designed shape $\Omega_\asD$ and an
as-manufactured shape $\Omega_\asM$. The method works for comparative
topological analysis (CTA) of {\it arbitrary} r-sets $\Omega_\asD, \Omega_\asM
\subseteq \R^\dimm$, representing the as-designed and (one possible)
as-manufactured shape. Note that the latter need {\it not} be one of the
parametric alternatives $\Omega_{\asM, \lambda}$ presented in Section
\ref{sec_model}. Our TCA can be applied for an arbitrary the transformation
$\Omega_\asD \mapsto \Omega_\asM$ (Fig. \ref{fig_manifolds}) that preserves
solidity and manifoldness \cite{Requicha1980representations}.

\begin{figure} [ht!]
	\centering
	\includegraphics[width=0.46\textwidth]{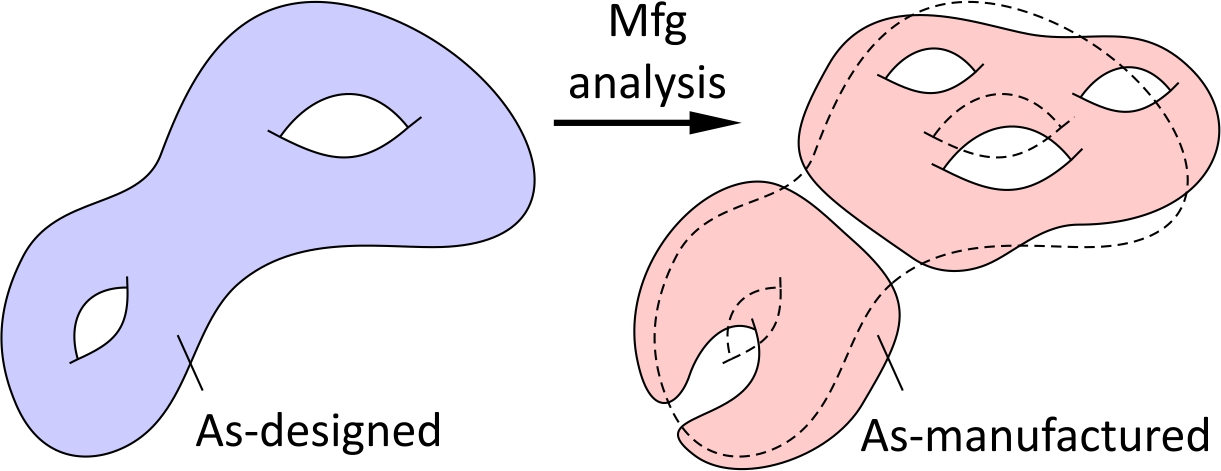}
	\caption{TCA applies to arbitrary changes from as-designed to as-manufactured
	solids. The only assumption in this section is solidity, i.e., both pointsets
	are compact-regular and semianalytic \cite{Requicha1980representations}.}
	\label{fig_manifolds}
\end{figure}

We quantify the AM-induced discrepancies in terms of the Euler characteristic
(EC) $\ec[\Omega]$ and the Betti numbers (BN) $\beti_j[\Omega]$ ($0 \leq i <
\dimm$). The two are related in $\dimm-$space by a simple linear combination
\cite{Hatcher2001algebraic}:
\begin{equation}
	\ec[\Omega] = \sum_{0 \leq j \leq \dimm} (-1)^j \beti_j[\Omega], \quad\Omega
	\subseteq \R^\dimm. \label{eq_EC_Betti}
\end{equation}
In 2D, $\beti_0[\cdot]$ and $\beta_1[\cdot]$ are the numbers of connected
components and holes, respectively. In 3D, $\beti_0[\cdot]$, $\beta_1[\cdot]$,
and $\beti_2[\cdot]$ are the numbers of connected components, tunnels, and
cavities, respectively. Intuitively, each $\beta_j[\cdot]$ counts the number of
qualitatively different $j-$dimensional closed surfaces contained within the
$\dimm-$dimensional shape, i.e., surfaces that cannot be transformed to one
another via continuous deformations without exiting the shape. Both EC and BN
are topological properties in the sense that they are invariant under continuous
deformations of the shape itself, which does not apply to the AM process in
general (viewed as a transformation $\Omega_\asD \mapsto \Omega_\asM$).

\subsection{Identifying Local Contributions}

Examining global topological changes by comparing the values of
$\ec[\Omega_\asD]$ (resp. $\beti_j[\Omega_\asD]$) with $\ec[\Omega_\asM]$ (resp.
$\beti_j[\Omega_\asM]$) provides little insight into what and how the different
geometric features contribute to the discrepancies. Our goal is to quantify
local topological discrepancies and identify what features are responsible for
them {\it with precise spatial information} that can be used for systematic
design and process correction. Figures \ref{fig_global_1}, \ref{fig_global_2},
and \ref{fig_global_3} illustrate a few possible scenarios in 2D when the global
analysis is insufficient (e.g., overall EC and/or BNs remain unchanged).

\begin{figure} [ht!]
	\centering
	\includegraphics[width=0.46\textwidth]{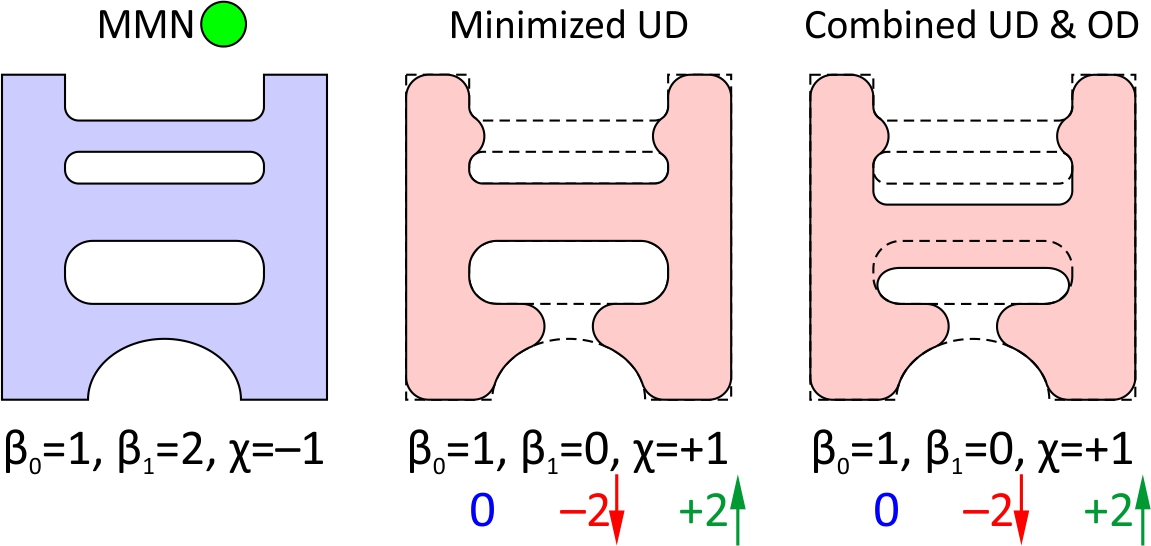}
	\vspace{-1.em}
	\caption{A 2D slice printed with the ``minimized UD'' policy
	\eq{eq_lambda_min_asM}. The broken two beams do not contribute to the total
	number of connected components, because the middle connection remains intact.}
	\label{fig_global_1}
	\vspace{1.0em}
	\centering
	\includegraphics[width=0.46\textwidth]{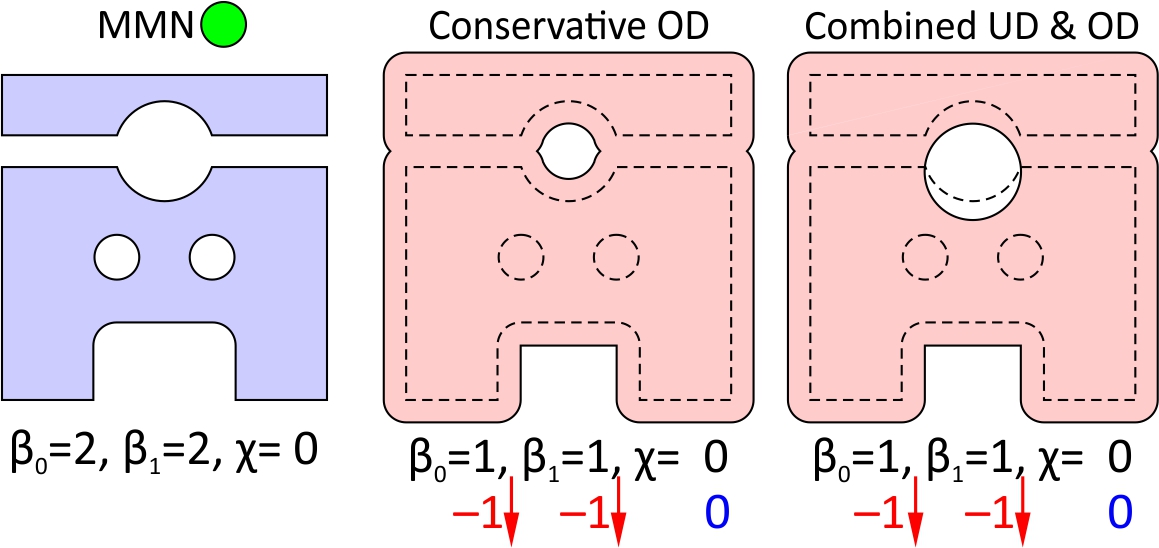}
	\vspace{-1.em}
	\caption{A 2D slice printed with the ``conservative OD'' policy
	\eq{eq_lambda_max_asM}. The two separate components are merged to one, the two
	holes are covered, and another one is formed, resulting in an unchanged EC.}
	\label{fig_global_2}
	\vspace{1.0em}
	\centering
	\includegraphics[width=0.46\textwidth]{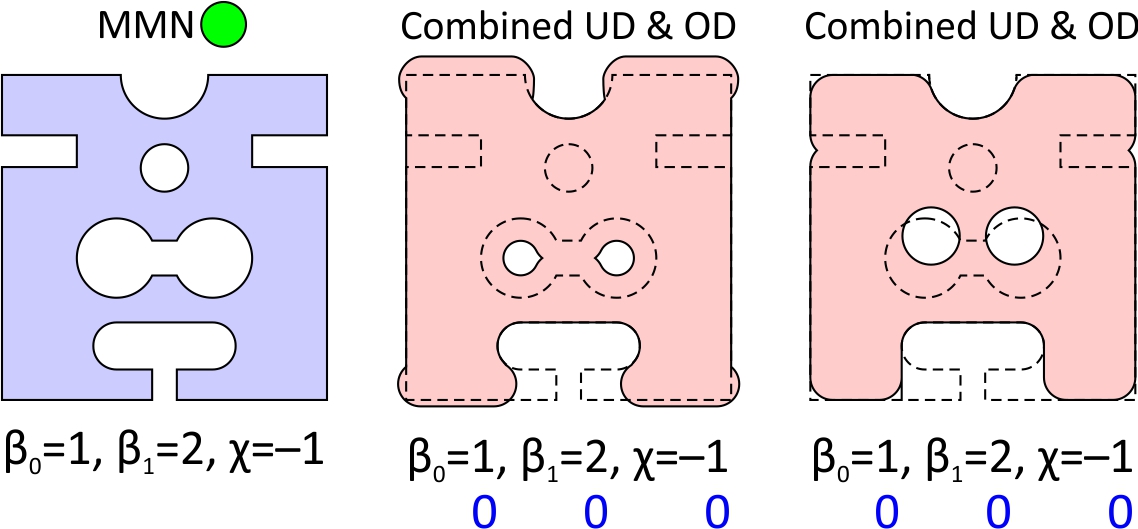}
	\vspace{-1.em}
	\caption{A more complex deposition policy than Figs. \ref{fig_global_1} and
	\ref{fig_global_2}. The connectivity does not change; one hole is covered while
	another is split into two, keeping both BNs and EC unchanged.}
\label{fig_global_3} \end{figure}

\parag{UD/OD Deviation ``Features''}

Let $C := (\Omega_\asD \cap^\ast \Omega_\asM)$ be the region of space that is
common between as-designed and as-manufactured shapes. The set differences
between these shapes amount to the under- and over-deposition (UD/OD) regions,
denoted respectively by $U, O \subseteq \R^\dimm$:
\begin{equation}
	U := (\Omega_\asD -^\ast \Omega_\asM), \quad\text{and}\quad O := (\Omega_\asM
	-^\ast \Omega_\asD). \label{eq_U_O}
\end{equation}
we need to carefully distinguishing regularized set operations (denoted via
asterisk) from ordinary set operations. Here, $\cap^\ast, -^\ast$ stand for
regularized set intersection and set difference \cite{tilove1980closure} which
ensure the resulting set is regularized (i.e., ``dangling'' surfaces/edges are
removed).%
\footnote{We do not need to worry about the set union operation, because the
	union of closed-regular sets is always closed-regular
	\cite{Requicha1980representations}. Hence, regularized set union $\cup^\ast$ can
	be replaced with ordinary set union $\cup$ when both participants are
	closed-regular sets (e.g., r-sets).}
The as-designed and as-manufactured shapes can be related set-theoretically as
follows:
\begin{align}
    \Omega_\asD &= (\Omega_\asD \cap^\ast \Omega_\asM) \cup (\Omega_\asD -^\ast
	\Omega_\asM) = (C \cup U), \label{eq_DM} \\
	\Omega_\asM &= (\Omega_\asM \cap^\ast \Omega_\asD) \cup (\Omega_\asM -^\ast
	\Omega_\asD) = (C \cup O), \label{eq_MD}
\end{align}
If the as-manufactured shape is reasonably close to the as-designed shape, the
OD/UD pointsets are made of small and scattered pieces, as we illustrate with
examples. Let us decompose the OD/UD pointsets into their respective (disjoint)
connected components (CC):
\begin{align}
	U &:= \bigcup_{1 \leq i \leq n_U} U_i, \quad (U_i \cap U_j) = \emptyset
	~~\text{if}~~ i \neq j, \label{eq_CC_U} \\
	O &:= \bigcup_{1 \leq i \leq n_O} O_i, \quad (O_i \cap O_j) = \emptyset
	~~\text{if}~~ i \neq j, \label{eq_CC_O}
\end{align}
Note that the CCs do not even touch along boundaries, hence their ordinary
pairwise intersections are empty.

\parag{Local Topology of the Features}

EC is an {\it additive} property, i.e., for every pair of (potentially
intersecting) pointsets $A, B$ with well-defined ECs, we have:
\begin{equation}
	\ec[A \cup B] = \ec[A] + \ec[B] - \ec[A \cap B], \label{eq_add}
\end{equation}
Note that even if the two shapes are closed-regular, the intersection $(A \cap
B)$ is not regularized, and can have lower-dimensional regions such as surface
patches, curve segments, and points. As long as the sets are semianalytic, these
features are well-behaved with computable ECs. Applying \eq{eq_add} to both
sides of \eq{eq_DM} and \eq{eq_MD} yields:
\begin{align}
	\ec[\Omega_\asD] &= \ec[C] + \ec[U] - \ec[C \cap U], \label{eq_DM_EC} \\
	\ec[\Omega_\asM] &= \ec[C] +  \ec[O] - \ec[C \cap O], \label{eq_MD_EC}
\end{align}
Subtracting \eq{eq_MD_EC} from \eq{eq_DM_EC} eliminates the EC $\ec[C]$ of the
common region $C = (\Omega_\asD \cap^\ast \Omega_\asM)$ and yields:
\begin{align}
	\ec[\Omega_\asM] - \ec[\Omega_\asD] &= \Big( \ec[O] - \ec[C \cap O] \Big)
	\nonumber \\
	&- \Big( \ec[U] - \ec[C \cap U] \Big), \label{eq_add_EC}
\end{align}
The terms $\ec[U]$ and $\ec[O]$ can be expanded further (using \eq{eq_add}) in
terms of the CCs of $U$ and $O$ in \eq{eq_CC_U} and \eq{eq_CC_O}:
\begin{align}
	\ec[U] &= \sum_{1 \leq i \leq n_U} \ec[U_i], \label{eq_U_EC} \\
	\ec[O] &= \sum_{1 \leq i \leq n_O} \ec[O_i], \label{eq_O_EC}
\end{align}
noting that the components are mutually exclusive, i.e., $(U_i \cap U_j) = (O_i
\cap O_j) = \emptyset$ hence $\ec[U_i \cap U_j] = \ec[O_i \cap O_j] =
\ec[\emptyset] = 0$ for every pair of indices $i \neq j$.

The remaining terms in \eq{eq_add_EC} are the ECs of the pointsets $(C \cap U)$
and $(C \cap O)$. These pointsets consist of lower-dimensional regions over
which the UD/OD regions intersect the common region $C = (\Omega^\ast_\asD
\cap^\ast \Omega^\ast_\asM)$. They comprise {\it portions} of the boundaries
$\partial \Omega_\asM$ and $\partial \Omega_\asD$, as illustrated in Fig.
\ref{fig_UDOD}. For r-sets, which are assumed to be semianalytic in solid
modeling \cite{Requicha1980representations}, these intersection terms can be
stratified into a combination of $2-$, $1-$, and $0-$strata, i.e., surface
patches, curve segments, or points, respectively.

\begin{figure} [ht!]
	\centering
	\includegraphics[width=0.48\textwidth]{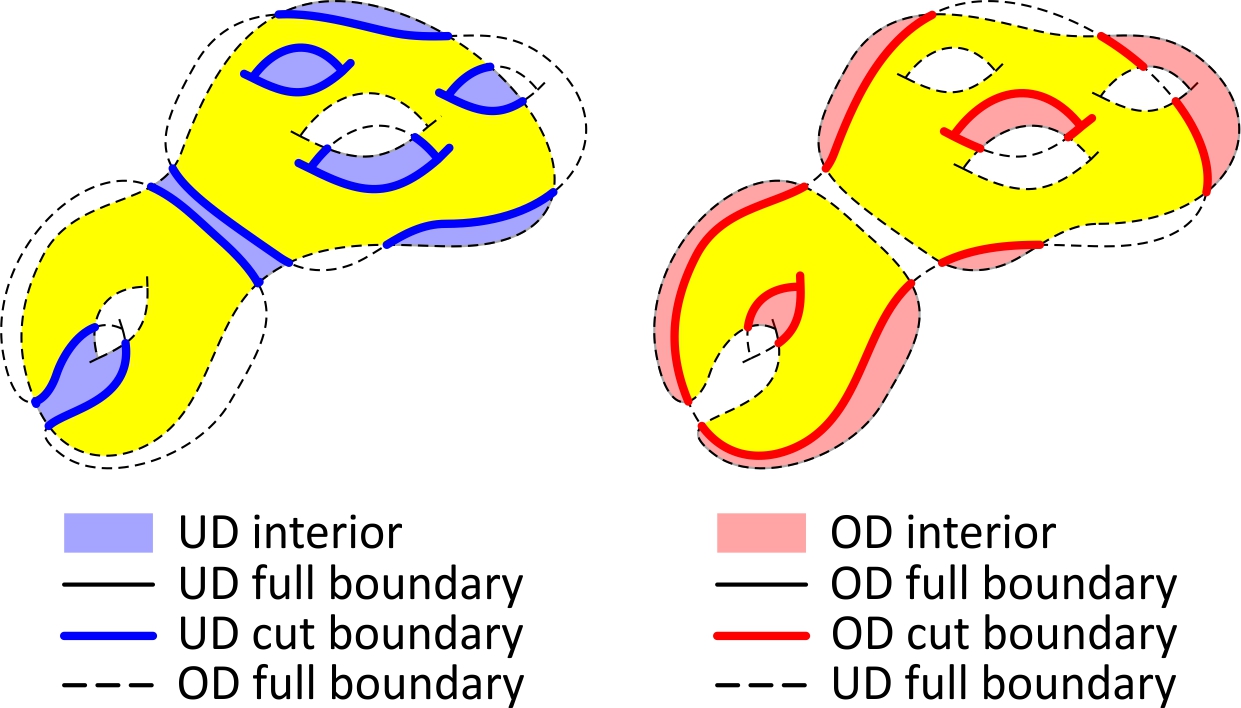}
	\caption{For arbitrary changes in the global topology of the disturbed shape,
	we identify the deviations in terms of contributions of local UD/OD features.
	In 3D, the interiors are volumetric and boundaries can be any combinations of
	surface patches, curve segments, and points, since the intersecting shapes are
	semianalytic.} \label{fig_UDOD}
\end{figure}

\begin{figure*}
	\centering
	\includegraphics[width=0.96\textwidth]{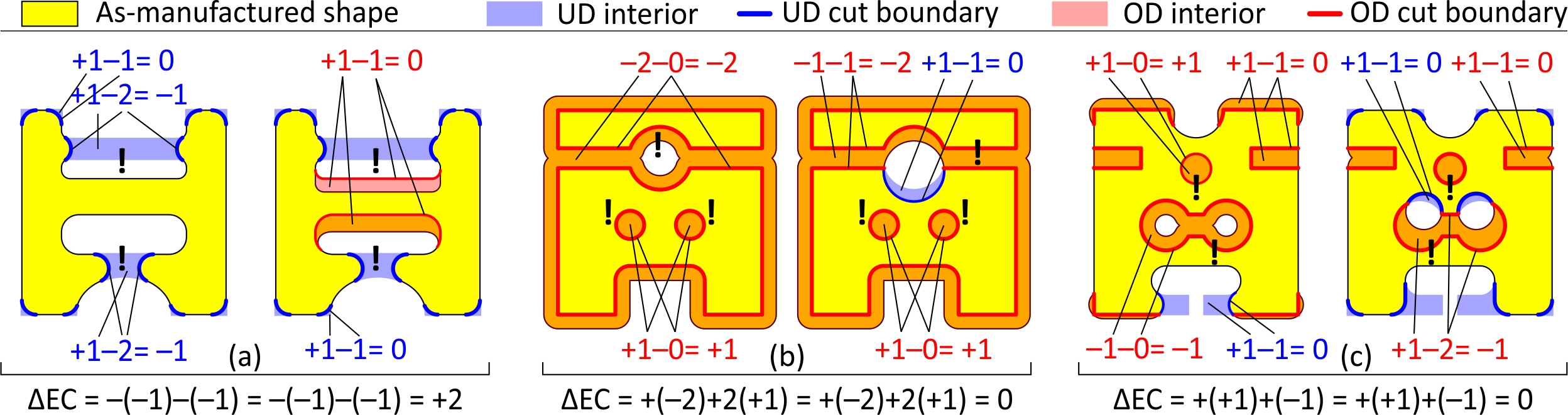}
	\caption{The three pairs of as-manufactured slices presented in Figs.
	\ref{fig_global_1}, \ref{fig_global_2}, and \ref{fig_global_3} are revisited
	using CTA. The contributions of each non-simple UD/OD feature (EC of its
	interior minus EC of its cut boundary) are shown. The total EC is obtained from
	accumulating these contributions with minus/plus signs for UD/OD, respectively.
	Even when the total sums to zero, CTA reveals the contributions of non-simple
	feature, so that they can be separated and corrected. {\it This information was
		lost in the global analysis} of Figs. \ref{fig_global_1}, \ref{fig_global_2},
	and \ref{fig_global_3}.} \label{fig_analysis}
\end{figure*}

\begin{lemma} (Euler Characteristic of Cut Boundaries)
	\begin{align}
		\ec[C \cap U] &= \sum_{1 \leq i \leq n_U}  \ec[C \cap \partial U_i],
		\label{eq_Ucutb_EC} \\
		\ec[C \cap O] &= \sum_{1 \leq i \leq n_O} \ec[C \cap \partial O_i],
		\label{eq_Ocutb_EC}
	\end{align}
\end{lemma}

\begin{proof}
	For every pair of solids $A$ and $B$ whose regularized set intersection is
	empty $(A \cap^\ast B) = \emptyset$, their ordinary set intersection can only
	occur over their common portion of boundaries, leading to the following
	identities:%
	\footnote{In measure-theoretic terms, $(A \cap^\ast B) = \emptyset
	~\rightleftharpoons~ \mu^\dimm[A \cap B] = 0$ which, in turn, can be expressed
	as an inner product $\langle \indic_A, \indic_B \rangle = 0$.}
	\begin{equation}
		(A \cap B) = (A \cap \partial B) = (\partial A \cap B) = (\partial A \cap
		\partial B), \label{eq_cut_id}
	\end{equation}
	which may or may not be empty. We refer to this pointset as the {\it cut
	boundary} of B with A (or vice versa) and denote it by $\partial_A B =
	\partial_B A$.
	
	The above identity is always true if we let $A := C$ and $B := U$ or $O$,
	noting that $(C \cap^\ast U) = (C \cap^\ast O) = \emptyset$ as a result of the
	definition in \eq{eq_U_O}. Hence:
	\begin{align}
		(C \cap U) &= (C \cap \partial U) = C \cap \partial \Big( \bigcup_{1 \leq i
		\leq n_U} U_i \Big), \\
		(C \cap O) &= (C \cap \partial O) = C \cap \partial \Big( \bigcup_{1 \leq i
		\leq n_O} O_i \Big).
	\end{align}
	Noting that the boundary of the union of disjoint CCs is the same as the union
	of their boundaries, we obtain:
	\begin{align}
		(C \cap U) &= C \cap \Big( \bigcup_{1 \leq i \leq n_U} \partial U_i \Big), \\
		(C \cap O) &= C \cap \Big( \bigcup_{1 \leq i \leq n_O} \partial O_i \Big).
	\end{align}
	Using distributivity of intersection over unions, we obtain:
	\begin{align}
		(C \cap U) &= \bigcup_{1 \leq i \leq n_U} (C \cap \partial U_i), \\
		(C \cap O) &= \bigcup_{1 \leq i \leq n_O} (C \cap \partial O_i),
	\end{align}
	Note that \eq{eq_cut_id} holds for $A := C$ and $B := U_i$ or $O_i$ as well,
	since $(C \cap^\ast U_i) = (C \cap^\ast O_i) = \emptyset$ for all $1 \leq i
	\leq n_U$ or $n_O$, respectively. Hence:
	\begin{align}
		(C \cap U) &= \bigcup_{1 \leq i \leq n_U} \partial_{C} U_i, \\
		(C \cap O) &= \bigcup_{1 \leq i \leq n_O} \partial_{C} O_i,
	\end{align}
	Applying \eq{eq_add} to both sides of the above equations and noting that the
	cut boundaries $\partial_{C} U_i$ and $\partial_{C} O_i$ are all mutually
	disjoint, we obtain:
	\begin{align}
		\ec[C \cap U] &= \sum_{1 \leq i \leq n_U}  \ec[\partial_{C} U_i],
		\label{eq_Ucutb_EC_} \\
		\ec[C \cap O] &= \sum_{1 \leq i \leq n_O} \ec[\partial_{C} O_i],
		\label{eq_Ocutb_EC_}
	\end{align}
	which is the same as \eq{eq_Ucutb_EC} and \eq{eq_Ocutb_EC}.
\end{proof}

The following key result follows from substituting the ECs from \eq{eq_U_EC},
\eq{eq_O_EC}, \eq{eq_Ucutb_EC}, and \eq{eq_Ocutb_EC} into \eq{eq_add_EC}:

\begin{coro} (Local Contributions to Global EC)
	\begin{align}
	\ec[\Omega_\asM] - \ec[\Omega_\asD] &= \sum_{1 \leq i \leq n_U} \Big( 
	\ec[O_i] - \ec[\partial_{C} O_i] \Big) \nonumber \\
	&-\sum_{1 \leq i \leq n_O} \Big( 
	\ec[U_i] - \ec[\partial_{C} U_i] \Big). 
	\label{eq_main}
	\end{align}
\end{coro}

The above result motives the following definition of a deviation ``feature'' for
AM processes:

\begin{defn} (Deviation Features)
A UD/OD `deviation feature' as a pair $\mathbb{F} := (F, \partial_C F)$, where
$F \subseteq \R^\dimm$ is its solid component and $\partial_C F = (F \cap
\partial C)$ is its cut boundary with another solid $C \subseteq \R^\dimm$ such
that $(C \cap^\ast F) = \emptyset$.

We can also define the notion of EC {\it contribution} (ECC) for features in
order to simplify \eq{eq_main}:
\begin{equation}
	\ecc[\mathbb{F}; C] := \ec[F] - \ec[\partial_C F] = \ec[F] - \ec[F \cap
	\partial C], \label{eq_def_EC_feat}
\end{equation}
The feature is called `simple' with respect to a set $C$ such that $(C \cap^\ast
F) = \emptyset$ (or $C-$simple for short) if $\ecc[\mathbb{F}; C] = 0$.
\end{defn}

To summarize, we proved that {\it the total change of EC due to AM errors can be
	obtained from a sum of contributions of ECs of individual UD/OD features}:
\begin{align}
	\ec[\Omega_\asM] - \ec[\Omega_\asD] = \sum_{1 \leq i \leq n_F} \pm
	\ecc[\mathbb{F}_i; C]. \label{eq_main_simple}
\end{align}
where $n_F = n_U + n_O$ is the total number of features. The sign $\pm$ depends
on the UD/OD type of each feature.

\begin{itemize}
	\item UD features are $\mathbb{U}_i := (U_i, \partial_{C}U_i)$ where $U_i$
	stands for the $i$\th CC of $(\Omega_\asD -^\ast \Omega_\asM)$ and
	$\partial_{C}U_i = (C \cap \partial U_i)$ is its cut boundary with $C =
	(\Omega_\asD \cap^\ast \Omega_\asM)$. They contribute $-\ecc[\mathbb{U}_i; C]$
	to the total change of EC.
	\item OD features are $\mathbb{O}_i := (O_i, \partial_{C}O_i)$ where $O_i$
	stands for the $i$\th CC of $(\Omega_\asM -^\ast \Omega_\asD)$ and
	$\partial_{C}O_i = (C \cap \partial O_i)$ is its cut boundary with $C =
	(\Omega_\asD \cap^\ast \Omega_\asM)$. They contribute $+\ecc[\mathbb{O}_i; C]$
	to the total change of EC.
\end{itemize}

\begin{coro}
	A feature $\mathbb{F}$ does not contribute to the topological change (in terms
	of EC) from $\Omega_\asD$ to $\Omega_\asM$ iff it is simple with respect to $C
	= (\Omega_\asD \cap^\ast \Omega_\asM)$.
\end{coro}

The above results provide a mechanism to order the UD/OD features in terms of
their topological significance. Each feature's ECC ($\ecc[\mathbb{U}_i; C]$ or
$\ecc[\mathbb{O}_i; C]$) to the global variation (i.e., the terms on the
right-hand side of \eq{eq_main}) is computed in parallel by subtracting the EC
of its cut boundary from the EC of its solid part. Simple features do not
contribute anything.

Figure \ref{fig_analysis} illustrates in 2D how the EC contributions are used to
distinguish between simple features such as a rounded corner and non-simple
features such as broken beams and filled holes. Even for situations in which
there is no change in the total EC (i.e., the left-hand side of \eq{eq_main} or
\eq{eq_def_EC_feat} vanished), the analysis reveals how different features
cancel each other's additive contributions. Each feature can thus be {\it
	corrected independently}, within the limitations of the AM process, to achieve
an as-manufactured shape that is both globally and locally equivalent to the
as-designed shape.

\subsection{Correcting Design and Process} \label{sec_correct}

Once the contributions of different UD/OD features to the overall EC are
determined, the deposition policy can be changed locally to modify the AM
instrument's motion and its resulting as-manufactured shape:
\begin{itemize}
	\item Simple UD/OD features do not contribute anything to the global EC (i.e.,
	$\ec[\mathbb{F}_i; C] = 0$), hence need not be corrected as far as the EC is of
	concern.
	\item If a UD feature contributes $-\ec[\mathbb{U}_i; C] \neq 0$, it should be
	brought back to $\Omega_\asM$ by including AM instrument configurations whose
	sweep of MMN deposits in that region, without affecting other regions.
	\item If an OD feature contributes $+\ec[\mathbb{O}_i; C] \neq 0$, it should be
	eliminated from $\Omega_\asM$ by excluding AM instrument configurations whose
	sweep of MMN deposits in that region, without affecting other regions.
\end{itemize}	
Common examples in 3D are (recalling \eq{eq_EC_Betti}):
\begin{itemize}
	\item A bridge whose solid component is simply-connected (i.e., $\ec[F_i] =
	1-0+0 = 1$) but connects locally disconnected regions in $C$ through $k \geq 2$
	simply-connected surfaces (i.e., $\ec[\partial_C F_i] = k(1-0+0) = k$). Hence
	we have $\ec[\mathbb{F}_i; C] = (1-k) \leq 1$.
	\item A tunnel whose solid component is simply-connected (i.e., $\ec[F_i] = 1$)
	but has a cut boundary with $C$ that is fully-cylindrical (i.e.,
	$\ec[\partial_C F_i] = 1-1+0 = 0$). Hence we have $\ec[\mathbb{F}_i; C] = (1-0)
	= +1$.
	\item A cavity whose solid component is simply-connected (i.e., $\ec[F_i] =1$)
	but has a cut boundary with $C$ that is fully spherical (i.e., $\ec[\partial_C
	F_i] = 1-0+1 = 2$). Hence we have $\ec[\mathbb{F}_i; C] = (1-2) = -1$.
\end{itemize} 
If these features were under- or over-deposited---i.e., either existed in
$\Omega_\asD$ but are missing from $\Omega_\asM$ (UD), or were not part of
$\Omega_\asD$ but appeared in $\Omega_\asM$ (OD)---they will contribute a
nonzero EC to the total EC.

Once the problematic (i.e., non-simple) features are identified and ranked based
on their ECCs, we can make  local adjustments to the AM instrument's motion,
such that its resulting sweep of MMN in \eq{eq_sweep} includes/excludes UD/OD
features of topological consequence.
\begin{itemize}
	\item For every non-simple UD feature $\mathbb{U}_i = (U_i, \partial_C U_i)$,
	we locally {\it decrease} the OMRT field in \eq{eq_lambda_field} by adjusting
	the weights, so that it includes configurations in the vicinity of the
	originally under-deposited region.
	\item For every non-simple UD feature $\mathbb{O}_i = (O_i, \partial_C O_i)$,
	we locally {\it increase} the OMRT field in \eq{eq_lambda_field} by adjusting
	the weights, so that it excludes configurations in the vicinity of the
	originally over-deposited region.
\end{itemize}
Every update of the OMRT field results in an updated AM instrument motion, hence
a new as-manufactured shape. Adding/removing non-simple UD/OD features can have
side effects on other UD/OD features. However, if the modifications are small
enough, the classification of features may not change dramatically and the model
can be corrected after several iterations. In the interest of keeping this paper
short and focused, we will discuss iterative design and process correction
elsewhere.

%% file: results.tex
\section{Results} \label{sec_results}

In this section, we show examples of applying the developed analysis of Sections
\ref{sec_model} and \ref{sec_UDOD} to complex 3D structures. The implementation
is in C++ (on \textsf{Linux}). All 3D models (except Fig. \ref{fig_helix_EC})
are obtained from \textsf{GrabCAD}
(\href{https://grabcad.com}{https://grabcad.com}) and \textsf{Thingiverse}
(\href{https://www.thingiverse.com/}{www.thingiverse.com}) and are visualized using
\textsf{ParaView} (\href{https://www.paraview.org/}{www.paraview.org}).

\begin{figure}[ht!]
	\centering
	\includegraphics[width=0.46\textwidth]{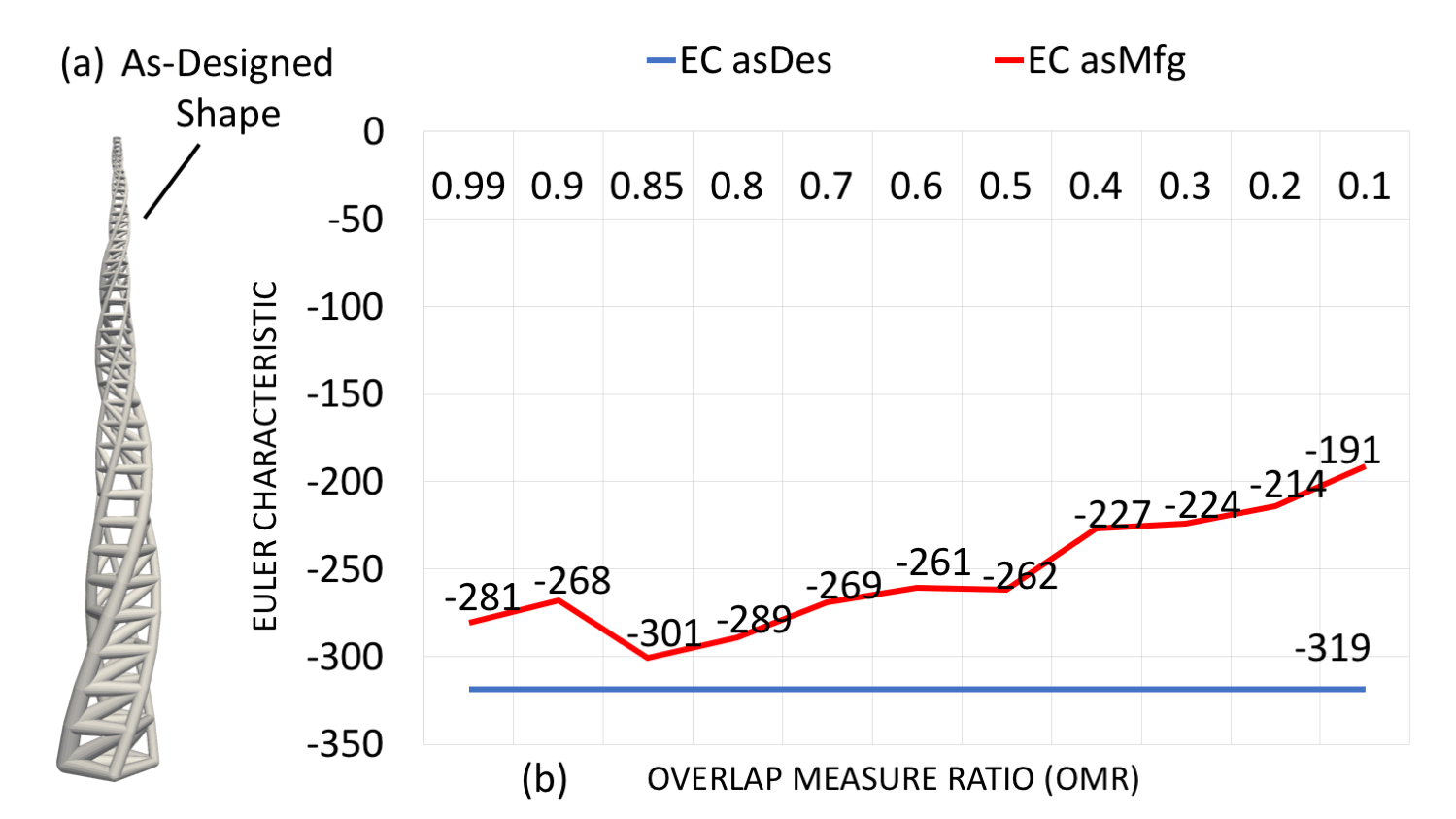}
	\caption{The effect of OMR on as-manufactured EC for a helical lattice
	structure (courtesy of \textsf{Siemens} Corporate Research) voxelized at $256
	\times 256 \times 2048$ using a spherical MMN of diameter $\sim 10$ voxels.}
\label{fig_helix_EC}
	\vspace{0.5em}
	\centering
	\includegraphics[width=0.46\textwidth]{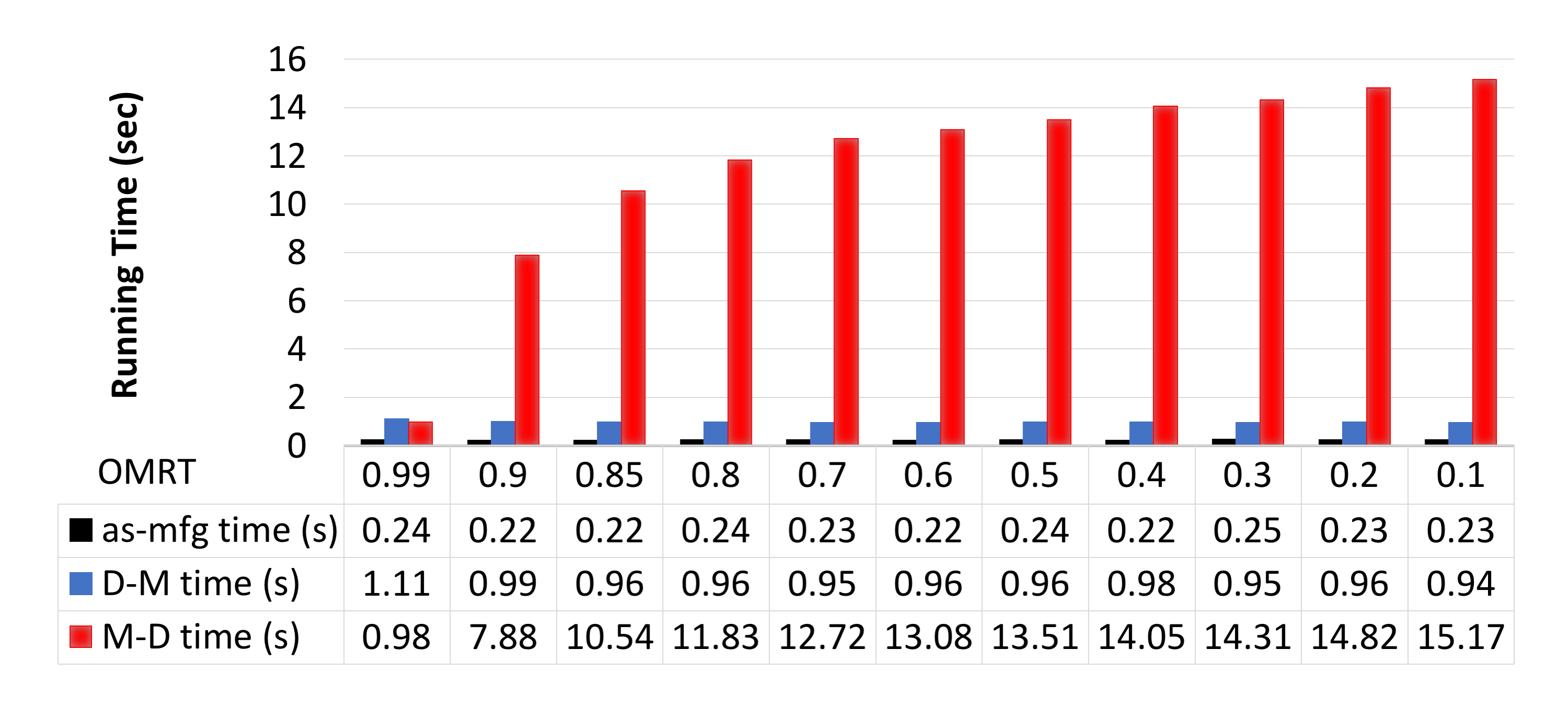}
	\caption{Computation times for as-manufactured model and local EC contributions
	of UD/OD features. The former requires two FFT-based convolutions and
	superlevel set operations, computed using \textsf{ArrayFire} on \textsf{NVIDIA
		GTX 1080} GPU (2,560 CUDA cores, 8GB RAM) via \textsf{OpenCL}. The latter was
	computed in parallel on \textsf{Intel Core}$^\textsf{TM}$ i7-7820X CPU @
	3.60GHz (8 cores, 32GB RAM) via \textsf{OpenMP}.} \label{fig_helix_time}
\end{figure}

\begin{figure*}
	\centering
	\includegraphics[width=0.96\textwidth]{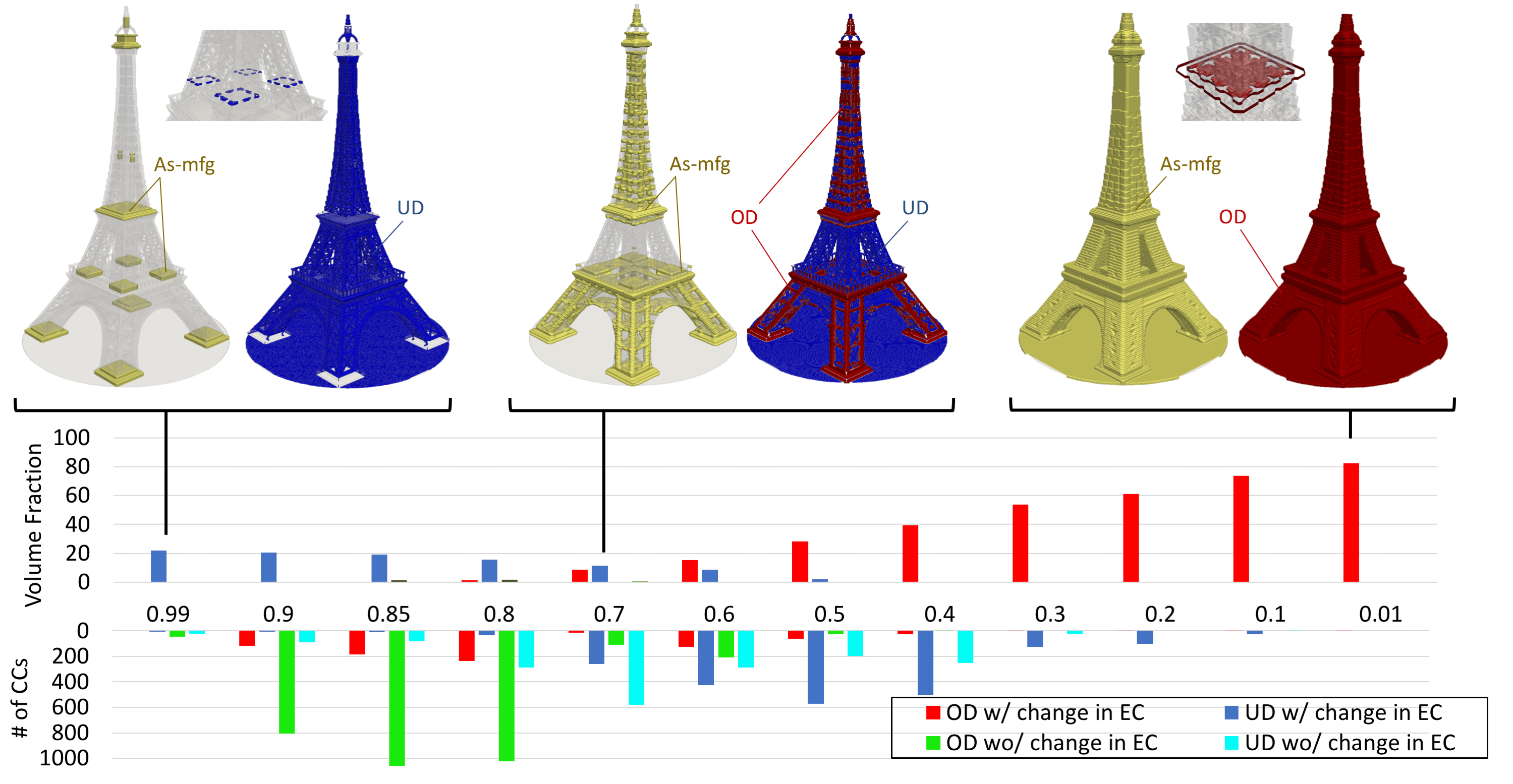}
	\caption{The effect of $\lambda$ on the volume fraction (VF) and number
	connected components (\#CC) for UD/OD features that do and do not contribute to
	the global EC (i.e., simple and non-simple features, respectively). The 3D model is obtained from \textsf{GrabCAD}.}
\label{fig_eiffel}
\end{figure*}

\begin{figure*}%\vspace*{-1cm}
	\centering
\includegraphics[width=0.95\textwidth]{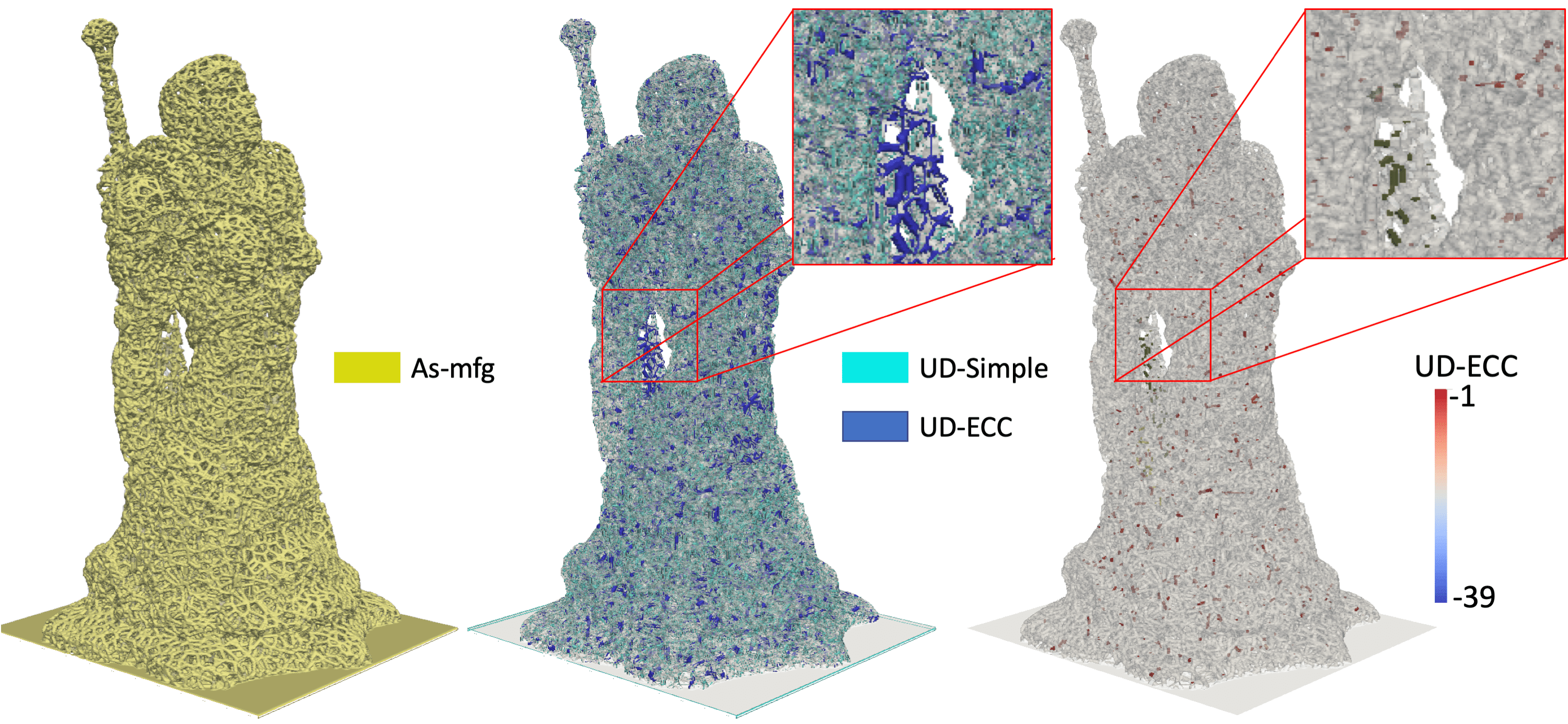}
\caption{Knight example from \textsf{Thingiverse}. (a) the US features with nonzero
	ECC (blue) and zero ECC (cyan); (b) the OD features with nonzero ECC (red) and
	zero ECC (green). The ECC field for UD features is shown in (c).}
\label{fig_knight}
\end{figure*}

\begin{table*}
	\centering
	\caption{The volume fraction (VF), number of connected components (\#CC), and
	computation time for UD/OD ECC computations. Note that the VF stands for the
	ratio of UD/OD volumes to the as-designed volume.}
	\begin{tabular}[t]{cccccccc}
		\hline \hline
		Example & $\lambda$ & resolution & UD VF (\#CC) & OD VF (\#CC)& as-mfg time
		(s) & UD time (s) & OD time (s) \\
		\hline
		Knight &$0.95$ &$256^2\times1024$ & 1.33 (6656) & 0 (0) & 0.12 & 3.92 & 0.48 \\ 
		Moomin &$0.95$ &$512^3$ & 10.3 (3659) & 0.63 (2903) & 0.28 & 6.36 & 5.35 \\ 
		Copter &$0.95$ &$512^3$ & 0.36 (1522) & 2.63 (3349) & 0.23 & 2.42 & 9.87 \\ 
		\hline
	\end{tabular}
	\label{tab:resultsSummary}
\end{table*}%

\begin{figure*} [ht!]
	\centering
	\includegraphics[width=0.96\textwidth]{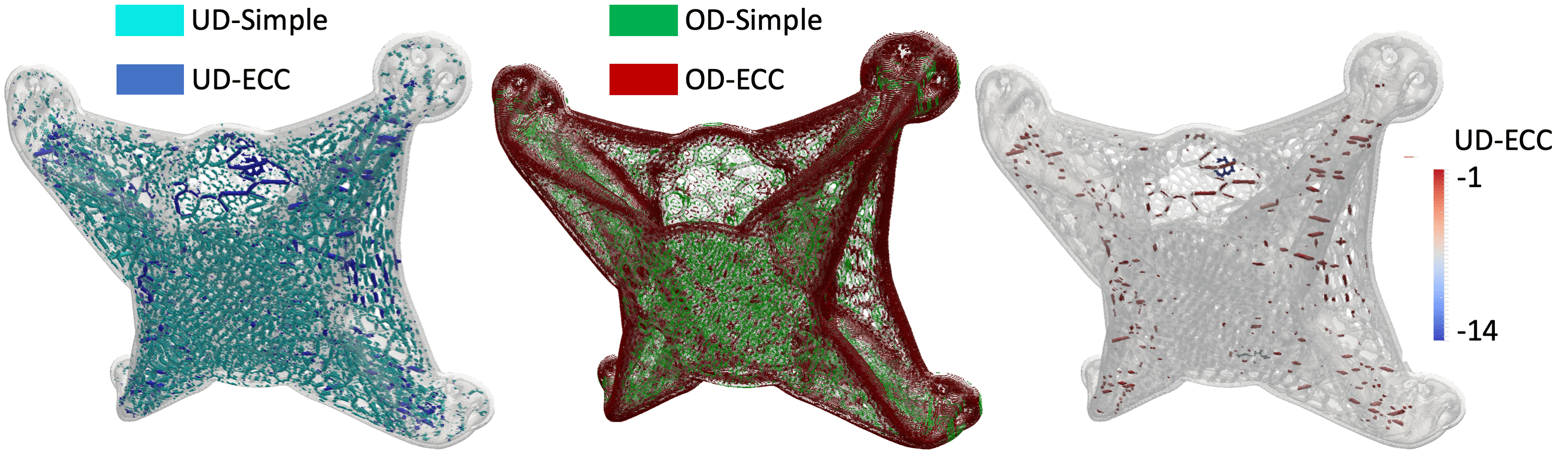}
	\includegraphics[width=0.96\textwidth]{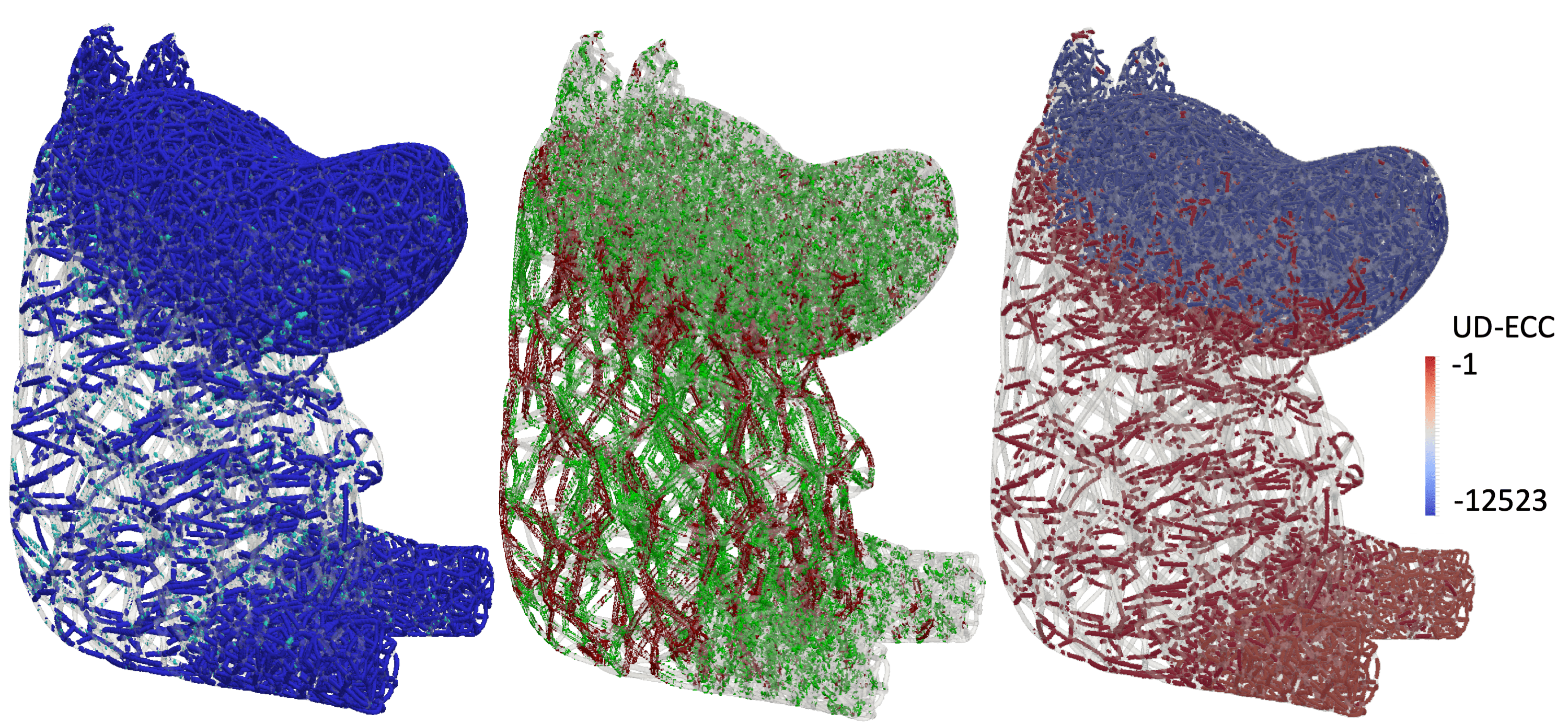}
	\caption{Copter example from \textsf{GrabCAD} (top) and Moomin example form \textsf{Thingiverse} (bottom). Left:
	the UD features with nonzero ECC (blue) and zero ECC (cyan). Middle: the OD
	features with nonzero ECC (red) and zero ECC (green). Right: the ECC field for
	UD features.} \label{fig_examps}
\end{figure*}

Figure \ref{fig_helix_EC} (a) shows a helical lattice in 3D with as-designed EC
of $\ec[\Sigma_\asD] = 1-320+0 = -319$. The global ECs for as-manufactured
variants $\Omega_{\asM, \lambda}$ obtained from the one-parametric formula in
\eq{eq_family_shape} are plotted in Fig. \ref{fig_helix_EC} (b) against the
global OMR threshold changing from $\lambda \to 1^-$ (minimized UD) to $\lambda \to
0^+$ (conservative OD). UD results in larger number of disconnected regions,
hence an increase in the (negative) EC. However, as the allowance for OD
increases by reducing the OMR, we have both UD/OD that alleviate the increase in
number of components while decreasing the number of tunnels with competing
effects on EC due to \eq{eq_EC_Betti}. We could apply persistent homology
\cite{Edelsbrunner2008persistent} to get a better understanding of the critical
OMR values for the birth/death of these topological features. However, the
global EC/BN plots or persistence graphs/barcodes lack spatial information
needed for design correction.

Figure \ref{fig_helix_time} illustrates the CPU times for parallel computation
of ECC for UD/OD features. Notice that as the allowance for over-deposition is
increased by decreasing the OMRT, the time increases almost linearly with the
overlap measure, which is proportional to the volume (and number of voxels) for
OD deviations from design.

The Eiffel tower is an interesting structure with features at different size
scales. Figure \ref{fig_eiffel} shows the effect of OMRT on the generation of
UD/OD features. Figures \ref{fig_knight} and \ref{fig_examps}
illustrate topological analysis on high-resolution geometries with intricate
interior lattices. Table \ref{tab:resultsSummary} summarizes the results for the
three shapes, including running times.

%% file: conclusions.tex
\section{Conclusion}

 A design's manufacturability (via an AM process) is largely determined by the
AM machine's ability to print the shape within `acceptable limits'. The notion
of geometric dimensioning and tolerancing has been used successfully to define
and check these limits for conventionally manufactured parts, but it is
challenging to define features of size for AM, and efforts are ongoing.
Nonetheless it is clear that combinations of manufacturing plans and 3D printer
resolutions will produce deviations between as-designed and as-manufactured
shapes, with no systematic procedure to check and control this deviation.

In this paper we have demonstrated an approach to topologically analyze and
classify the  deviations between as-designed and as-manufactured shapes. The
approach uses fundamental properties of the Euler characteristic to quantify the
additive contributions of disjoint regions of over/under-deposition. Ambiguities
of global topological analysis arising from the canceling effects of these
regions on the local topological changes are resolved. Important geometric
features that appear/disappear or are otherwise deformed are identified. The
precise spatial information is used to provide remedial measures to retain the
geometric feature and the local topology.

Our comparative topological analysis (CTA) in Section \ref{sec_UDOD} is
independent of any particular method of computing as-manufactured shapes, hence
is not limited to the measure-theoretic model presented in Section
\ref{sec_model}. Different AM processes may require domain-specific simulation
methods that take into account additional process constraints such as
solidification, build orientation, support structure, chamber temperature, and
so on.

There are several possible extensions to this paper. Our approach considers a
single as-manufactured outcome at a time. It provides little insight on the
persistence of local changes across the spectrum of as-manufactured variants
that one can obtain by changing the 3D printer specs or deposition
policies---e.g., by changing the shape/size of MMN, threshold on OMR, etc.

As opposed to analyzing a given as-manufactured model for topological
consistency, we may alternatively try to construct a process plan that considers
topological changes and produces a part that will (by construction) be
topologically equivalent to the as-designed shape.

Solving this problem will also help address broader challenges in design for
additive manufacturing (DfAM).

\section*{Acknowledgement}

This research was developed with funding from the Defense Advanced Research
Projects Agency (DARPA). The views, opinions and/or findings expressed are those
of the authors and should not be interpreted as representing the official views
or policies of the Department of Defense or U.S. Government.

%% file: topolMfg.bbl
\begin{thebibliography}{10}
\expandafter\ifx\csname url\endcsname\relax
  \def\url#1{\texttt{#1}}\fi
\expandafter\ifx\csname urlprefix\endcsname\relax\def\urlprefix{URL }\fi
\expandafter\ifx\csname href\endcsname\relax
  \def\href#1#2{#2} \def\path#1{#1}\fi

\bibitem{Robbins2016efficient}
J.~Robbins, S.~J. Owen, B.~W. Clark, T.~E. Voth, An efficient and scalable
  approach for generating topologically optimized cellular structures for
  additive manufacturing, Additive Manufacturing 12 (2016) 296--304.

\bibitem{Haertel2017fully}
J.~H.~K. Haertel, G.~F. Nellis, A fully developed flow thermofluid model for
  topology optimization of {3D}-printed air-cooled heat exchangers, Applied
  Thermal Engineering 119 (2017) 10--24.

\bibitem{Balic2006intelligent}
J.~Balic, Intelligent {CAD/CAM} systems for {CNC} programming---an overview,
  Advances in Production Engineering \& Management 1~(1) (2006) 13--22.

\bibitem{Nouri2016structural}
H.~Nouri, S.~Guessasma, S.~Belhabib, Structural imperfections in additive
  manufacturing perceived from the {X}-ray micro-tomography perspective,
  Journal of Materials Processing Technology 234 (2016) 113--124.

\bibitem{Wu2016self}
J.~Wu, X.~Wang, C. C. L .and~Zhang, R.~Westermann, Self-supporting rhombic
  infill structures for additive manufacturing, Computer-Aided Design 80 (2016)
  32--42.

\bibitem{Wu2018infill}
J.~Wu, N.~Aage, R.~Westermann, O.~Sigmund, Infill optimization for additive
  manufacturing---approaching bone-like porous structures, IEEE Transactions on
  Visualization and Computer Graphics 24~(2) (2018) 1127--1140.

\bibitem{Martinez2016procedural}
J.~Mart{\'\i}nez, J.~Dumas, S.~Lefebvre, Procedural {V}oronoi foams for
  additive manufacturing, ACM Transactions on Graphics (TOG) 35~(4) (2016) 44.

\bibitem{Martinez2017orthotropic}
H.~Mart{\'\i}nez, J.and~Song, J.~Dumas, S.~Lefebvre, Orthotropic k-nearest
  foams for additive manufacturing, ACM Transactions on Graphics (TOG) 36~(4)
  (2017) 121.

\bibitem{Cajal2013volumetric}
C.~Cajal, J.~Santolaria, J.~Velazquez, S.~Aguado, J.~Albajez, Volumetric error
  compensation technique for {3D} printers, Procedia Engineering 63 (2013)
  642--649.

\bibitem{Tong2008error}
K.~Tong, S.~Joshi, E.~Amine~L., Error compensation for fused deposition
  modeling ({FDM}) machine by correcting slice files, Rapid Prototyping Journal
  14~(1) (2008) 4--14.

\bibitem{Telea2011voxel}
A.~Telea, A.~Jalba, Voxel-based assessment of printability of {3D} shapes, in:
  International Symposium on Mathematical Morphology and Its Applications to
  Signal and Image Processing, Springer, 2011, pp. 393--404.

\bibitem{Ibrahim2009dimensional}
D.~Ibrahim, T.~L. Broilo, C.~Heitz, M.~G. de~Oliveira, H.~W. de~Oliveira,
  S.~M.~W. Nobre, J.~H.~G. dos Santos~Filho, D.~N. Silva, Dimensional error of
  selective laser sintering, three-dimensional printing and
  polyjet$^\textrm{TM}$ models in the reproduction of mandibular anatomy,
  Journal of Cranio-Maxillofacial Surgery 37~(3) (2009) 167--173.

\bibitem{Rao2016assessment}
P.~K. Rao, Z.~Kong, C.~E. Duty, R.~J. Smith, V.~Kunc, L.~J. Love, Assessment of
  dimensional integrity and spatial defect localization in additive
  manufacturing using spectral graph theory, Journal of Manufacturing Science
  and Engineering 138~(5) (2016) 051007.

\bibitem{Silva2008dimensional}
D.~N. Silva, M.~G. De~Oliveira, E.~Meurer, M.~I. Meurer, J.~V.~L. da~Silva,
  A.~Santa-B{\'a}rbara, Dimensional error in selective laser sintering and
  {3D}-printing of models for craniomaxillary anatomy reconstruction, Journal
  of Cranio-Maxillofacial Surgery 36~(8) (2008) 443--449.

\bibitem{wang2016improved}
W.~M. Wang, C.~Zanni, L.~Kobbelt, Improved surface quality in {3D} printing by
  optimizing the printing direction, in: Computer Graphics Forum, Vol.~35,
  Wiley Online Library, 2016, pp. 59--70.

\bibitem{Mahesh2004benchmarking}
M.~Mahesh, Y.~S. Wong, J.~Y.~H. Fuh, H.~T. Loh, Benchmarking for comparative
  evaluation of {RP} systems and processes, Rapid Prototyping Journal 10~(2)
  (2004) 123--135.

\bibitem{Arrieta2012quantitative}
C.~Arrieta, S.~Uribe, J.~Ramos-Grez, A.~Vargas, P.~Irarrazaval, V.~Parot,
  C.~Tejos, Quantitative assessments of geometric errors for rapid prototyping
  in medical applications, Rapid Prototyping Journal 18~(6) (2012) 431--442.

\bibitem{Cabiddu2017maps}
D.~Cabiddu, M.~Attene, $\epsilon-$maps: Characterizing, detecting and
  thickening thin features in geometric models, Computers \& Graphics 66 (2017)
  143--153.

\bibitem{Nelaturi2015manufacturability}
S.~Nelaturi, W.~Kim, T.~Kurtoglu, Manufacturability feedback and model
  correction for additive manufacturing, Journal of Manufacturing Science and
  Engineering 137~(2) (2015) 021015.

\bibitem{Nelaturi2015representation}
S.~Nelaturi, V.~Shapiro, Representation and analysis of additively manufactured
  parts, Computer-Aided Design 67 (2015) 13--23.

\bibitem{Rosen2018inferring}
P.~Rosen, M.~Hajij, J.~Tu, T.~Arafin, L.~Piegl, Inferring quality in point
  cloud-based {3D} printed objects using topological data analysis, arXiv
  preprint arXiv:1807.02921.

\bibitem{Edelsbrunner2008persistent}
H.~Edelsbrunner, J.~Harer, Persistent homology--a survey, Contemporary
  Mathematics 453 (2008) 257--282.

\bibitem{Liu2015identification}
S.~Liu, Q.~Li, W.~Chen, L.~Tong, G.~Cheng, An identification method for
  enclosed voids restriction in manufacturability design for additive
  manufacturing structures, Frontiers of Mechanical Engineering 10~(2) (2015)
  126--137.

\bibitem{Li2016structural}
Q.~Li, W.~Chen, S.~Liu, L.~Tong, Structural topology optimization considering
  connectivity constraint, Structural and Multidisciplinary Optimization 54~(4)
  (2016) 971--984.

\bibitem{Attali2009stability}
D.~Attali, J.~D. Boissonnat, H.~Edelsbrunner, Stability and computation of
  medial axes--a state-of-the-art report, in: Mathematical Foundations of
  Scientific Visualization, Computer Graphics, and Massive Data Exploration,
  Springer, 2009, pp. 109--125.

\bibitem{Behandish2018automated}
M.~Behandish, S.~Nelaturi, J.~de~Kleer, Automated process planning for hybrid
  manufacturing, Computer-Aided Design 102 (2018) 115--127.

\bibitem{Nelaturi2019exploring}
S.~Nelaturi, A.~M. Mirzendehdel, M.~Behandish, Exploring feasible design spaces
  for heteroheneous constraints, 2019, special Issue on Generative Design.

\bibitem{Lysenko2010group}
M.~Lysenko, S.~Nelaturi, V.~Shapiro, Group morphology with convolution
  algebras, in: Proceedings of the 14th ACM Symposium on Solid and Physical
  Modeling (SMP'2010), ACM, 2010, pp. 11--22.

\bibitem{Requicha1980representations}
A.~A.~G. Requicha, Representations for rigid solids: Theory, methods, and
  systems, ACM Computing Surveys (CSUR) 12~(4) (1980) 437--464.

\bibitem{Selig2005geometrical}
J.~M. Selig, Geometrical Fundamentals of Robotics, Springer-Verlag New York,
  2005.

\bibitem{Behandish2017analytic}
M.~Behandish, Analytic methods for geometric modeling, {Ph.D.} dissretation,
  University of Connecticut (2017).

\bibitem{Kavraki1995computation}
L.~Kavraki, Computation of configuration-space obstacles using the fast
  {F}ourier transform, IEEE Transactions on Robotics and Automation 11~(3)
  (1995) 408--413.

\bibitem{Hatcher2001algebraic}
A.~Hatcher, Algebraic topology, Cornell University, 2001.

\bibitem{tilove1980closure}
R.~Tilove, A.~A. Requicha, Closure of boolean operations on geometric entities,
  Computer-Aided Design 12~(5) (1980) 219--220.

\end{thebibliography}
